\documentclass[11pt]{article}

\usepackage{amsmath,amssymb,amsthm,amsfonts,latexsym,bbm,xspace,graphicx,float,mathtools}
\usepackage[backref,colorlinks,citecolor=blue,bookmarks=true]{hyperref}
\usepackage[letterpaper,margin=1in]{geometry}
\usepackage{color}
\usepackage{comment}
\usepackage{algorithmic,algorithm}

\usepackage{tweaklist}
\renewcommand{\enumhook}{\setlength{\topsep}{2pt}%
  \setlength{\itemsep}{-2pt}}
\renewcommand{\itemhook}{\setlength{\topsep}{2pt}%
  \setlength{\itemsep}{-2pt}}

\usepackage{fullpage}
\usepackage{appendix}
\usepackage[nomessages]{fp}

\usepackage{bbm}

\newtheorem{theorem}{Theorem}[section]

\newtheorem{lemma}{Lemma}[section]

\newtheorem{claim}{Claim}

\newtheorem{fact}{Fact}
\newtheorem{definition}{Definition}[section]

\newtheorem{observation}{Observation}

\newtheorem{remark}{Remark}[section]

\newcommand{\R}{\ensuremath{\mathbb{R}}} 

\newcommand{\poly}{{\rm poly}}

\newcommand{\notshow}[1]{{}}

\newcommand{\1}{\mathbbm{1}}

\DeclareMathOperator{\E}{\mathbb{E}}

\def \SW  {{\cal SW^*}}

\def\Ber{{\mathcal{B}}}

\DeclareMathOperator{\argmax}{argmax}

\definecolor{MyGray}{rgb}{0.8,0.8,0.8}

\begin{document}
\title{On Simultaneous Two-player Combinatorial Auctions}

\author{
Mark Braverman \thanks{Department of Computer Science, Princeton University, email: mbraverm@cs.princeton.edu. Research supported in part by an NSF CAREER award (CCF-1149888), NSF CCF-1215990, NSF CCF-1525342, a Packard Fellowship in Science and Engineering, and the Simons Collaboration on Algorithms and Geometry.}
\and
Jieming Mao  \thanks{Department of Computer Science, Princeton University, email: jiemingm@cs.princeton.edu.}
\and 
S. Matthew Weinberg \thanks{Department of Computer Science, Princeton University, email: smweinberg@princeton.edu. Research completed in part while the author was a Research Fellow at the Simons Institute for the Theory of Computing. }
}
\addtocounter{page}{-1}
\maketitle
\begin{abstract} 
We consider the following communication problem: Alice and Bob each have some valuation functions $v_1(\cdot)$ and $v_2(\cdot)$ over subsets of $m$ items, and their goal is to partition the items into $S, \bar{S}$ in a way that maximizes the \emph{welfare}, $v_1(S) + v_2(\bar{S})$. We study both the \emph{allocation problem}, which asks for a welfare-maximizing partition and the \emph{decision problem}, which asks whether or not there exists a partition guaranteeing certain welfare, for binary XOS valuations. For interactive protocols with $\poly(m)$ communication, a tight 3/4-approximation is known for both~\cite{Feige06,DobzinskiS06}. 

For interactive protocols, the allocation problem is provably \emph{harder} than the decision problem: any solution to the allocation problem implies a solution to the decision problem with one additional round and $\log m$ additional bits of communication via a trivial reduction. Surprisingly, the allocation problem is provably \emph{easier} for simultaneous protocols. Specifically, we show:
\begin{itemize}
\item There exists a simultaneous, randomized protocol with polynomial communication that selects a partition whose expected welfare is at least $3/4$ of the optimum. This matches the guarantee of the best interactive, randomized protocol with polynomial communication. 
\item For all $\varepsilon > 0$, any simultaneous, randomized protocol that decides whether the welfare of the optimal partition is $\geq 1$ or $\leq 3/4 - 1/108+\varepsilon$ correctly with probability $> 1/2 + 1/ \poly(m)$ requires exponential communication. This provides a separation between the attainable approximation guarantees via interactive ($3/4$) versus simultaneous ($\leq 3/4-1/108$) protocols with polynomial communication.
\end{itemize}
In other words, this trivial reduction from decision to allocation problems provably requires the extra round of communication. We further discuss the implications of our results for the design of truthful combinatorial auctions in general, and extensions to general XOS valuations. In particular, our protocol for the allocation problem implies a new style of truthful mechanisms. 


\end{abstract}
\newpage


\section{Introduction}\label{sec:intro}
\vspace{-2mm}
Intuitively, search problems (find the optimal solution) are considered ``strictly harder'' than decision problems (does a solution with quality $\geq Q$ exist?) for the following (formal) reason: once you find the optimal solution, you can simply evaluate it and check whether its quality is $\geq Q$ or not. The same intuition carries over to approximation as well: once you find a solution whose quality is within a factor $\alpha$ of optimal, you can distinguish between cases where solutions with quality $\geq Q$ exist and those where all solutions have quality $\leq \alpha Q$. The easy conclusion one then draws is that the communication (resp. runtime) required for an $\alpha$-approximation to any decision problem is upper bounded by the communication (resp. runtime) required for an $\alpha$-approximation to the corresponding search problem plus the communication (resp. runtime) required to evaluate the quality of a proposed solution. 

Note though that for communication problems, in addition to the negligible increase in communication (due to evaluating the quality of the proposed solution), this simple reduction might also require (at least) an extra round of communication (because the parties can evaluate a solution's quality only after it is found). Still, it seems hard to imagine that this extra round is really necessary, and that somehow protocols exist that guarantee an (approximately) optimal solution without (approximately) learning their quality. The surprising high-level takeaway from our main results is that \emph{this extra round of communication is provably necessary}: Theorems~\ref{thm:mainupper} and~\ref{thm:mainlower} provide a natural communication problem (combinatorial auctions) such that a $3/4$-approximation for the search problem can be found by a simultaneous protocol\footnote{A simultaneous protocol has one round of communication: Alice and Bob each simultaneously send a message and then no further communication takes place.} with polynomial communication, but every simultaneous prootcol guaranteeing a ($3/4-1/108 + \varepsilon$)-approximation for the decision problem requires $\exp(m)$ communication.

At this point, we believe our results to have standalone interest, regardless of how we wound up at this specific communication problem. But there is a rich history related to the design of truthful combinatorial auctions motivating our specific question, which we overview below.
\vspace{-1mm}
\subsection{Combinatorial Auctions - how did we get here?}

In a combinatorial auction, a designer with $m$ items wishes to allocate them to $n$ bidders so as to maximize the \emph{social welfare}. That is, if bidder $i$ has a monotone valuation function $v_i: 2^{[m]} \rightarrow \mathbb{R}_+$,\footnote{By monotone, we mean that $v_i(S) \geq v_i(T)$ for all $T \subseteq S$.} the designer wishes to find disjoint sets $S_1,\ldots, S_n$ maximizing $\sum_i v_i(S_i)$. The history of combinatorial auctions is rich, and the problem has been considered with and without incentives, with and without Bayesian priors, and in various models of computation (see Appendix~\ref{sec:related} for brief overview). The overarching theme in all of these works is to try and answer the following core question: \emph{Are truthful mechanisms as powerful as (not necessarily truthful) algorithms}?

For many instantiations of the above question, the answer is surprisingly yes. For example, without concern for computational/communication complexity, the celebrated Vickrey-Clarke-Groves auction is a truthful mechanism that always selects the welfare-maximizing allocation (and therefore achieves welfare equal to that of the best algorithm)~\cite{Vickrey61, Clarke71, Groves73}. Of course, the welfare maximization problem is NP-hard and also requires exponential communication between the bidders, even to guarantee a $1/\sqrt{m}$-approximation. A poly-time algorithm (with polynomial communication) is known to match this guarantee~\cite{Raghavan1988, KolliopoulosS1998, BriestKV2005}, and interestingly, a poly-time truthful mechanism (with polynomial communication) was later discovered as well~\cite{LaviS05}. 

The state of affairs gets even more interesting if we restrict to proper subclasses of monotone valuations such as \emph{submodular} valuations.\footnote{A function is submodular if $v(S) + v(T) \geq v(S \cup T) + v(S \cap T)$.} Here, a very simple greedy algorithm is known to find a 1/2-approximation in both $\poly(n, m)$ black-box value queries to each $v_i(\cdot)$, and polynomial runtime (in $n, m$, and the description complexity of each $v_i(\cdot)$)~\cite{LehmannLN01}, and a series of improvements provides now a (1-1/e)-approximation, which is tight~\cite{Vondrak08, MirrokniSV08, DobzinskiV12}. Yet, another series of works also proves that any truthful mechanism that runs in polynomial time (in $n, m$, and the description complexity of each $v_i(\cdot)$), or makes only $\poly(n, m)$ black-box value queries to each $v_i(\cdot)$ achieves at best an $1/m^{\Omega(1)}$-approximation~\cite{PapadimitriouSS08, BuchfuhrerDFKMPSSU10, BuchfuhrerSS10,DanielySS15, Dobzinski11,DughmiV11, DobzinskiV12b}. So while poly-time algorithms, or algorithms making $\poly(n,m)$ black-box value queries can achieve constant-factor approximations, poly-time truthful mechanisms and truthful mechanisms making $\poly(n, m)$ black-box value queries can only guarantee an $1/m^{\Omega(1)}$-approximation, and there is a separation.

But this is far from the whole story: already ten years ago, quite natural truthful mechanisms were developed that achieved an $1/O(\log^2 m)$-approximation~\cite{DobzinskiNS06}, which were subsequently improved to $1/O(\sqrt{\log m})$~\cite{Dobzinski07, KrystaV13, Dobzinski16a}, and even hold for the much broader class of XOS valuations.\footnote{A valuation is XOS if there exists a matrix of item valuations $v_{ij}$ and $v_i(S) = \max_{j}\{\sum_{i \in S} v_{ij}\}$. XOS valuations are also called fractionally subadditive, and are a proper subclass of subadditive valuations (where $v(S \cup T) \leq v(S) + v(T)$).} As these approximation guarantees are better than the lower bounds referenced in the previous paragraph, it seems that perhaps there should be some kind of contradiction: any reasonable definition of ``natural'' should imply ``poly-time,'' right? The catch is that each of these mechanisms are essentially posted-price mechanisms: they (essentially) offer each bidder a price $p_j$ for item $j$, and let the buyer choose any subset of items they want to purchase. These prices can be computed in poly-time, but the barrier is that deciding which subset of items the bidder wishes to purchase, called a \emph{demand query}, is in general NP-hard (assuming a succinct representation of the valuation function is given), or requires exponentially many black-box value queries. So the only reason these mechanisms don't fall victim to the strong lower bounds of the previous paragraph is because they get to ask each bidder to compute a single demand query, and this query is used to select exactly the set of items that bidder receives. 

The point is that while these existing separations are major results, and rule out certain classes of natural truthful mechanisms from achieving desirable approximation ratios, they are perhaps not addressing ``the right'' model if posted-price mechanisms with poly-time computable prices provide approximation guarantees that significantly outperform known lower bounds. Therefore, it seems that communication is really the right complexity measure to consider, if one wants the resulting lower bounds to hold against all ``natural'' mechanisms. Unfortunately, the state-of-affairs for communication complexity of combinatorial auctions lags pretty far behind the aforementioned complexity measures. For instance, existing literature doesn't provide a single lower bound against truthful mechanisms that doesn't also hold against algorithms. That is, wherever it's known that no truthful mechanism with communication at most $C$ obtains an approximation ratio better than $\alpha$ when buyers have valuations in class $V$, it's because it's also known that no algorithm/protocol with communication at most $C$ obtains an approximation ratio better than $\alpha$ when buyers have valuations in class $V$. On the other hand, the best known truthful mechanisms with polynomial communication for (say) XOS bidders achieve an $1/O(\sqrt{\log m})$ approximation~\cite{Dobzinski16a}, while the best known algorithms with polynomial communication obtain a $(1-1/e)$-approximation ~\cite{DobzinskiS06,Feige06,FeigeV06}. Even for the case of just two bidders, the best known truthful mechanisms with polynomial communication achieve a $1/2$-approximation (which is trivial - just give the grand bundle of all items to whoever values it most), while the best known algorithms with polynomial communication achieve a $3/4$-approximation (which is tight). It's fair to say that determining whether or not there's a separation in what approximation guarantees are possible for algorithms with polynomial communication and truthful mechanisms with polynomial communication for any class of valuations between submodular and subadditive is one of the core, concrete open problems in Algorithmic Mechanism Design.

Progress on this front had largely been stalled until very recent work of Dobzinski provided a clear path to possibly proving a separation (and it seems to be an accepted conjecture that indeed a separation exists)~\cite{Dobzinski16b}. Without getting into details of the complete result, one implication is the following: if there exists a truthful mechanism with polynomial communication for 2-player combinatorial auctions with XOS (/submodular/subadditive) valuations that guarantees an approximation ratio of $\alpha$, then there exists a \emph{simultaneous} protocol with polynomial communication for 2-player combinatorial auctions with XOS (/submodular/subadditive) valuations that guarantees an approximation ratio of $\alpha$ as well. Let us emphasize this point again: in general, interactive protocols with polynomial communication \emph{do not} imply simultaneous protocols with polynomial communication, and numerous well-known problems have polynomial interactive protocols, but require exponential simultaneous communication~\cite{Papadimitriou:1982:CC:800070.802192, Duris:1984:LBC:800057.808668, Nisan:1993:RCC:152322.152386, DobzinskiNO14, AlonNRW15, Assadi17}. But, Dobzinski's result asserts that because of the extra conditions on \emph{truthful} (interactive) mechanisms, their existence indeed implies a simultaneous (not necessarily truthful) protocol of comparable communication complexity. So ``all'' one has to do to prove lower bounds against truthful mechanisms for 2-player combinatorial auctions is prove lower bounds against simultaneous protocols, motivating the study of simultaneous 2-player combinatorial auctions.

At first glance, it perhaps seems obvious that achieving strictly better than a $1/2$-approximation via a simultaneous protocol should be impossible, and it's just a matter of finding the right tools to prove it.\footnote{Indeed, that is what the authors conjectured at the onset of this work.} This is because quite strong lower bounds are known for ``sketching'' valuation functions, that is, finding a succinct representation of a function that allows for \emph{approximate} evaluation of value queries. For example, it's known that any sketching scheme for XOS valuations that allows for evaluation of value queries to be accurate within a $o(m)$-factor requires superpoly($m$) size~\cite{BadanidiyuruDFKNR12}. So if somehow a $1/(2-\varepsilon)$-approximation could be guaranteed with a $\poly (m)$-communication simultaneous protocol, it is \emph{not} because enough information is transmitted to evaluate value queries within any non-trivial error. At first glance, it perhaps seems unlikely that such a protocol can possibly exist. Surprisingly, our work shows not only that a $1/(2-\varepsilon)$-approximation is achievable with $\poly(m)$ simultaneous communication, but (depending on exactly the question asked) $\poly(m)$ simultaneous communication suffices to achieve the same approximation guarantees as the best possible interactive protocol with $\poly(m)$ communication. 
\vspace{-1mm}
\subsection{Simultaneous Protocols for Welfare Maximization}
In this work, we specifically study the welfare maximization problem for two bidders with \emph{binary XOS valuations}.\footnote{A function is binary XOS if all $v_{ij}$ in the matrix representation are $0$ or $1$.} Binary XOS valuations are a natural starting point since welfare maximization is especially natural when phrased as a communication problem. Depending on whether one wants to decide the quality of the welfare-optimal allocation, or actually find an allocation inducing the optimal welfare, welfare maximization for binary XOS bidders is equivalent to one of the following:\footnote{Equivalent definitions are given in Section~\ref{sec:prelim} which are stated more in the language of welfare maximization. We pose these statements here since these formulations make for an especially natural communication problem.}
\begin{definition}[BXOS Decision Problem] Alice is given as input $a$ subsets of $[m]$, $A_1,\ldots, A_a$. Bob is given as input $b$ subsets of $[m]$, $B_1,\ldots, B_b$, and both see input $X$. Determine whether or not there exists an $i, j$ such that $|A_i \cup B_j| \geq X$. A protocol is an $\alpha$-approximation if whenever there exists an $i, j$ such that $|A_i \cup B_j| \geq X$, it answers yes, and whenever $\max_{i, j} \{|A_i \cup B_j|\} < X/\alpha$ it answers no, but may have arbitrary behavior in between.
\end{definition}

\begin{definition}[BXOS Allocation Problem] Alice is given as input $a$ subsets of $[m]$, $A_1,\ldots, A_a$. Bob is given as input $b$ subsets of $[m]$, $B_1,\ldots, B_b$. Output a partition of items $S, \bar{S}$ maximizing $\max_{i, j}\{|A_i \cap S| + |B_j \cap \bar{S}|\}$ (over all partitions).\footnote{A protocol is an $\alpha$-approximation if it outputs a partition $S, \bar{S}$ guaranteeing $\alpha \cdot \max_{i, j}\{|A_i \cap S| + |B_j \cap \bar{S}|\} \geq \max_{i, j, T}\{|A_i \cap T| + |B_j \cap T|\}$.}
\end{definition}

Recall that typically we think of decision problems as being ``easier'' than allocation/search problems: certainly if you can find a welfare maximizing allocation, you can also determine its welfare (and this claim is formal for interactive protocols with $\poly(m)$ communication). Our main result asserts that this intuition breaks down for simultaneous protocols: the decision problem is strictly harder than the allocation/search problem. To the best of our knowledge, this is the first instance of such a separation. 

\begin{theorem}\label{thm:mainupper} There exists a randomized, simultaneous protocol with $\poly(m)$ communication that obtains a $3/4$-approximation for the BXOS allocation problem. This is the best possible, as even randomized, interactive protocols require $2^{\Omega(m)}$ communication to do better. 
\end{theorem}

\begin{theorem}\label{thm:mainlower} For all $\varepsilon > 0$, any randomized, simultaneous protocol that obtains a $(3/4-1/108+\varepsilon)$-approximation for the BXOS decision problem {with probability larger than $1/2+1/\poly(m)$} requires $2^{\Omega(m)}$ communication. 
\end{theorem}
Future sections contain more precise versions (that reference the protocols achieving them) of Theorems~\ref{thm:mainupper} (Theorem~\ref{thm:main}) and~\ref{thm:mainlower} (Theorem~\ref{thm:lb}).

\subsection{Extensions and Implications for Truthful Combinatorial Auctions}
Part of the analysis of our protocols actually makes use of the binary assumption (as opposed to holding for general XOS). Part of the analysis, however, does not. In particular, our same protocols when applied to general XOS functions yield a deterministic, simultaneous $(3/4-1/32-\varepsilon)$-approximation for both problems, and a deterministic 2-round $(3/4-\varepsilon)$-approximation for both problems for general XOS functions. 

We are also able to show that a modification of our protocol yields a $1/2$-approximation for any number of binary XOS bidders, and that this protocol implies a \emph{strictly} truthful mechanism.\footnote{By this we mean it is a \emph{strongly} dominant strategy for bidders to follow the protocol, and not just that they are indifferent between following and not following.} The mechanism is quite different from existing approaches, and could inspire better truthful mechanisms in domains where previous molds (such as VCG-based/maximal-in-range) provably fail~\cite{PapadimitriouSS08, BuchfuhrerDFKMPSSU10,BuchfuhrerSS10, DanielySS15}. Essentially, the designer offers a menu of lotteries to each bidder and the cost of each lottery depends on how ``flexible'' the option is. So for instance, taking item one deterministically will be more expensive than taking a single item uniformly at random. The pricing scheme is designed exactly so that each bidder is strictly incentivized to follow our simultaneous protocol.  

Finally, while our results have standalone merit outside the scope of truthful combinatorial auctions, it is important to properly quantify their impact in this direction. Dobzinski's recent reduction shows that truthful combinatorial auctions with polynomial communication imply simultaneous algorithms for the \emph{allocation problem}. So Theorem~\ref{thm:mainlower} \emph{does not} rule out the possibility of a truthful mechanism for two XOS bidders that requires polynomial communication and guarantees a $3/4$-approximation (more on this in Section~\ref{sec:future}).


\subsection{Brief Preliminaries and Roadmap}
\label{sec:prelim}
Below we give some brief preliminaries. Section~\ref{sec:toy} provides a toy setting to help develop intuition for where the gap between allocation and decision problem comes from. Section~\ref{sec:warmup} provides a warmup for our protocols via a 2/3-approximation for the allocation problem and a 3/5-approximation for the decision problem. Sections~\ref{sec:ske} and~\ref{sec:ub} contain our positive results, and Section~\ref{sec:lb} contains details on our lower bound.

We'll use the terms ``protocol'' and ``mechanism'' to distinguish between cooperative players and strategic players. It will be clear from context exactly our model of communication, what we mean by approximation, truthfulness, etc., and a more formal treatment of this appears in Appendix~\ref{app:formalprelim}. We'll also use the notation $\SW(v_1,...,v_n) = \max_{S_1 \cap \cdots \cap S_n = \emptyset} \sum_{i=1}^n v_i(S_i)$ to denote the maximum social welfare attainable for the valuation profile $(v_1,\ldots,v_n)$.

Below are definitions of the valuation classes used in the paper. These are equivalent to the definitions used in Section~\ref{sec:intro}, but more apt for proofs and less apt for posing easy-to-parse communication problems.

\begin{definition}We consider the following classes of valuations:

\begin{itemize}
\item A valuation function $v$ is \textbf{additive} if for every bundle $S$, $v(S) = \sum_{i \in S} v(\{ i \})$. 
\item A valuation function $v$ is \textbf{XOS} if there exist additive valuations $a_1,...,a_t$ such that for every bundle $S$, $v(S) = \max_{i=1}^t a_i(S)$. Each $a_i$ is called a \textbf{clause} of $v$.
\item A valuation function $v$ is \textbf{binary additive} if $v$ is additive and for every item $i$, $v(\{i\})\in \{0,1\}$. We will sometimes refer to a binary additive valuation as a \textbf{set}, referring to $\{i | v(\{i\}) = 1\}$. 
\item A valuation function $v$ is \textbf{binary XOS} if $v$ is XOS and all $v$'s clauses are binary additive valuations. Again, we will sometimes refer to $v$'s clauses as \textbf{sets} to make it more natural to talk about unions/intersections/etc.
\end{itemize}
\end{definition}

\section{Intuition for the Gap: an Extremely Toy Setting}\label{sec:toy}
Consider the following very toy setting: Alice and Bob each have some valuation function $v(\cdot)$ such that $v([m]) \in [1,M]$, and $v(\cdot)$ is monotone (no other assumptions).\footnote{If one wishes, one could further restrict attention to submodular, XOS, etc., but this section is just supposed to be a toy model to provide some intuition, and we will not belabor this point.}

\begin{observation} In the very toy setting, Alice and Bob can guarantee the following tight approximation guarantees with zero communication:
\begin{itemize}
\item A $1/2$-approximation for the allocation problem with a randomized protocol: give all the items either to Alice or Bob uniformly at random. 
\item A $1/(M+1)$-approximation for the allocation problem with a deterministic protocol: give all the items to Alice.
\item A $1/(2M)$-approximation for the decision problem (decide if social welfare $\geq X$ or $\leq X/(2M)$, arbitrary behavior allowed in-between): If $X > 2M$ output ``$\leq X/(2M)$'' If $X \leq 2M$, ``$\geq X$.''
\end{itemize}
\end{observation}
Since this example is just to provide intuition, we omit a complete proof. The first bullet should be fairly clear: the optimal welfare is clearly upper bounded by $v_1([m])+v_2([m])$, and the protocol guarantees exactly half of this. The third bullet should also be clear: the optimal welfare is always between $1$ and $2M$. Moreover, any value in the range is possible ($2M$ if, for instance, $v_1(\{1\}) = M = v_2(\{2\})$. $1$ if, for instance, $v_1(S) =v_2(S) = 1$ iff $S \ni 1$, and $v_1(S) = v_2(S) = 0$ otherwise). So with zero communication, better than $1/(2M)$ is not possible. The middle bullet is perhaps the only tricky one. If we give all of the items to Alice, we guarantee welfare $v_1([m])\geq 1$, and the optimum is upper bounded by $v_1([m]) + M$. 

Again, the purpose of this example is just to provide intuition as to where this gap might come from, and we do not consider it a ``result.'' Of course, one should not expect the gaps to stay quite so drastic as we dial up the communication: with just $\log M$ bits in the above example, a deterministic protocol for the allocation problem and decision problem can both guarantee a $1/2$-approximation (output $v([m])$). But this example still captures some of the intuition as to where the gap comes from.


\vspace{-1mm}
\section{Warmup: Beating a $1/2$-Approximation}
\vspace{-2mm}
\label{sec:warmup}
Before explaining our protocol, consider the following thought experiment: say instead Alice and Bob are asked to just report a single clause from their valuation. What clause should they choose and how well will this protocol solve the allocation/decision problem? It's not too hard to see that the best they can do is to just report the largest clause in their list (maximizes $b_i([m])$ over all clauses $b_j$), which will obtain just a 1/2-approximation for each problem. Now, what if they each report \emph{two} clauses from their valuations, can they do something more clever? Well, they should certainly try to report clauses that are large, as this lets the other know which sets they value the most. But they should also try to report clauses that are different, as this allows for more flexibility in an allocation that both parties value highly. It's perhaps not obvious what the right tradeoff is between large/different (or even exactly what ``different'' should formally mean), but it turns out that a good approach is for Alice and Bob to each output the two clauses in their list with the largest union (i.e. output $b_i, b_j$ maximizing $\SW(b_i,b_j)$). Subject to figuring out how to translate this information into solutions, a slight variant of this protocol guarantees a $2/3$-approximation for the allocation problem, and a $3/5$-approximation for the decision problem, and the proof is actually quite simple. Note below that Theorem~\ref{thm:warmup1} holds only for BXOS, whereas Theorem~\ref{thm:warmup2} holds for general XOS. We'll provide both proofs below first, followed by a brief discussion.

\begin{algorithm}
	\floatname{algorithm}{Protocol}
        \caption{Simultaneous randomized warmup protocol for 2-party combinatorial auctions with binary XOS valuations}
        \label{pro:warmup1}
    \begin{algorithmic}[1]
	\STATE Alice finds $b_1,b_2,b_3$  among clauses of her valuation $v_1$ such that $b_1$ maximizes $b_1([m])$ and $b_2,b_3$ maximize $\SW(b_2,b_3)$. Then she picks $j$ uniformly at random from $\{1,2,3\}$ and sends $b_j$ to the auctioneer.
	\STATE For each item $i$, the auctioneer allocates it to Alice if $b_j(\{i\}) =1$; otherwise allocate it to Bob. 
             \end{algorithmic}
\end{algorithm}

\begin{theorem}
\label{thm:warmup1}
Protocol \ref{pro:warmup1} gives a $2/3$-approximation to the 2-party BXOS allocation problem. 
\end{theorem}

\begin{proof}
First, we want to claim that if Alice sents $b_j$ to the auctioneer, then the resulting welfare is at least $\SW(b_j,v_2)$. This is actually an instantiation of a claim we will want to reference later, so we'll state a more general form below:
\begin{claim}\label{claim:binaryadditive}
Let $b_1$ be a binary additive valuation and $v_2$ be a binary XOS valuation. Then the allocation that awards to Alice all items such that $b_1(\{i\}) = 1$ achieves welfare equal to $\SW(b_1,v_2)$. 
\end{claim}
\begin{proof} Let $A$ denote the set of items for which $b_1(\{i\}) = 1$, and consider any other allocation $(B, \bar{B})$. We first reason that we can remove from $B$ all items $\notin A$ without hurting $b_j(B) + v_2(\bar{B})$. This is trivial to see, as $b_j$ has value $0$ for all items $\notin A$. Next, we reason that we can add to $B$ any item $\in A$ without hurting $b_j(B) + v_2(\bar{B})$. To see this, observe that we are certainly increasing $b_j(B)$ by $1$ when we make this change, as $b_j$ is just additive and $b_j(\{i\}) = 1$ for all $i \in A$. In addition, we can't possibly decrease $v_2(\bar{B})$ by more than $1$, as all of the clauses in $v_2$ are binary additive (and therefore have value at most $1$ for any item). So again, the total change is only positive. At the end of these changes, observe that we have now transitioned from $(B, \bar{B})$ to $(A, \bar{A})$ without losing any welfare, and therefore $(A, \bar{A})$ is indeed optimal. 
\end{proof}

Claim~\ref{claim:binaryadditive} immediately lets us conclude that the expected welfare guaranteed by Protocol~\ref{pro:warmup1} is at least $\frac{1}{3} \cdot \sum_{j=1}^3 \SW(b_j,v_2)$. Now, let $S$ and $T$ be the optimal allocation to achieve $\SW(v_1,v_2)$. Let $a$ be the clause of $v_1$ such that $a(S) = v_1(S)$. Let $a'$ be the clause of $v_2$ such that $a'(T) = v_2(T)$. So $\SW(v_1,v_2) = a(S) + a'(T)$. From the protocol, we know that $b_1([m]) \geq a([m]) \geq a(S) $. Moreover, if $U$ and $U'$ are the allocation that achieves $\SW(b_2,b_3)$, then we know that $b_2(U) + b_3(U') = \SW(b_2,b_3) \geq \SW(a,b_1) \geq a(S) + b_1(T)$ (by definition of $b_2, b_3$). In expectation, the social welfare we get in the protocol is at least:
\begin{eqnarray*}
\frac{1}{3} \cdot \sum_{j=1}^3 \SW(b_j,v_2) &\geq& \frac{1}{3} \cdot \left( b_1(S) + a'(T)  + b_2(U) + a'(U') + b_3(U') + a'(U) \right) \\
&\geq&  \frac{1}{3} \cdot \left( b_1(S) + a'(T)  + a(S) + b_1(T) + a'([m]) \right) \geq \frac{1}{3} \cdot \left( b_1([m]) + a(S)  + 2a'(T) \right) \\
&\geq& \frac{1}{3} \cdot \left( 2a(S)  + 2a'(T) \right) = \frac{2}{3} \cdot \SW(v_1,v_2).
\end{eqnarray*}
\end{proof}

\begin{algorithm}
	\floatname{algorithm}{Protocol}
        \caption{Simultaneous deterministic warmup protocol for 2-party combinatorial auctions with XOS valuations}
        \label{pro:warmup2}
    \begin{algorithmic}[1]
	\STATE Alice finds $b_1,b_2,b_3$  among clauses of her valuation $v_1$ such that $b_1$ maximizes $b_1([m])$ and $b_2,b_3$ maximize $\SW(b_2,b_3)$. Bob finds $b_4,b_5,b_6$  among clauses of his valuation $v_2$ such that $b_4$ maximizes $b_4([m])$ and $b_5,b_6$ maximize $\SW(b_5,b_6)$. Alice sends $b_1,b_2,b_3$ to the auctioneer and Bob sends $b_4,b_5,b_6$ to the auctioneer simultaneously.
	\STATE \textbf{For allocation problem:} Auctioneer finds $j\in \{1,2,3\},j'\in \{4,5,6\}$  that maximizes $\SW(b_j,b'_{j'})$ and allocate items according to it. 
	\STATE \textbf{For decision problem:} Let $X$ be the parameter in the decision problem. Auctioneer finds $j\in \{1,2,3\},j'\in \{4,5,6\}$ that maximizes $\SW(b_j,b'_{j'})$. If $\SW(b_j,b'_{j'}) \geq 3X/5$, say "yes" ($\SW(v_1,v_2) \geq X$). If $\SW(b_j,b'_{j'}) < 3X/5$, say "no". 
             \end{algorithmic}
\end{algorithm}

\begin{theorem}\label{thm:warmup2}
Protocol \ref{pro:warmup2} gives a $3/5$-approximation to the 2-party XOS allocation problem and the 2-party XOS decision problem.\footnote{{XOS allocation problem and XOS decision problem are the obvious extensions of BXOS allocation problem and BXOS decision problem for non-binary clauses.}}
\end{theorem}

\begin{proof}
Let $S$ and $T$ be the optimal allocation to achieve $\SW(v_1,v_2)$. Let $a$ be the clause of $v_1$ such that $a(S) = v_1(S)$. Let $a'$ be the clause of $v_2$ such that $a'(T) = v_2(T)$. So $\SW(v_1,v_2) = a(S) + a'(T)$. From the protocol, we know that $b_1([m]) \geq a([m]) \geq a(S) $ and $b_4([m]) \geq a'([m]) \geq a'(T)$. Let $U$ and $U'$ be the allocation to achieve $\SW(b_2,b_3)$. We know that $b_2(U) + b_3(U') \geq a(S) + b_1(T)$. Let $W$ and $W'$ be the allocation to achieve $\SW(b_5,b_6)$. We know that $b_5(W) + b_6(W') \geq a'(T) + b_4(S)$. Then we have
\begin{eqnarray*}
\SW(b_1,b_5) + \SW(b_1,b_6) &\geq& b_1(W') +b_5(W) +b_1(W) + b_6(W') \\
&\geq& b_1([m]) + b_5(W) +b_6(W') \geq a(S) + a'(T) + b_4(S).
\end{eqnarray*}
Similarly we have
\[
\SW(b_2,b_4) + \SW(b_3,b_4)  \geq a(S) + a'(T) + b_1(T).
\]

The social welfare we get in the protocol is at least
\begin{eqnarray*}
&&\SW(b_j,b'_{j'}) \\
&\geq&\frac{1}{5} \cdot \left( \SW(b_1,b_4) + \SW(b_1,b_5) + \SW(b_1,b_6) +\SW(b_2,b_4) + \SW(b_3,b_4) \right) \\ 
&\geq& \frac{1}{5} \cdot \left( b_1(S) +b_4(T) + 2a(S) + 2a'(T) + b_4(S) + b_1(T) \right) \\
&\geq&  \frac{1}{5} \cdot \left( b_1([m]) +b_4([m]) + 2a(S) + 2a'(T)  \right) \geq \frac{3}{5} \left(a(S) + a'(T)\right)= \frac{3}{5} \SW(v_1,v_2).\\
\end{eqnarray*}
From this, it is easy to check that Protocol \ref{pro:warmup2} gives a $3/5$-approximation to both the 2-party XOS allocation problem and the 2-party XOS decision problem.
\end{proof}

So now there are two remaining questions: first, how does one generalize the reasoning in Protocols~\ref{pro:warmup1} and~\ref{pro:warmup2} to multiple clauses? And second, why the heck is there a difference between their guarantees for the allocation and decision problem for binary XOS valuations? For the first question, we'll postpone the details to Section~\ref{sec:ske}, but just note here that our full protocols indeed makes use of similar reasoning. For the second, observe that Claim~\ref{claim:binaryadditive} is somewhat magical: if Alice's valuation is binary additive, and Bob's is binary XOS, then it is possible to allocate the items optimally \emph{without any input from Bob} (other than the knowledge that his valuation is indeed binary XOS). While it's not obvious that Claim~\ref{claim:binaryadditive} should necessarily be quite so helpful (given that we do, in fact, get input from Bob), this turns out to be the crucial difference between the allocation and decision problem. At a high level, there is necessarily some information lost between Alice's valuation and her message (ditto for Bob). The decision problem requires us to deal with both losses, but Claim~\ref{claim:binaryadditive} lets certain kinds of protocols only worry about the loss from Alice.

\vspace{-1mm}
\section{Developing Good Sketches}
\vspace{-2mm}
\label{sec:ske}
In this section, we define ``sketches'' in some specific forms for binary XOS valuations.\footnote{Please forgive our slight abuse of the term ``sketch.'' We are using the term in the streaming sense, where ``sketch'' just refers to a concise but lossy representation of a larger input, and not specifically to proper sketches of valuation functions that preserve value queries within a constant factor.} They are the main ingredients in our protocols and mechanisms. At a high level, the sketches are trying to simultaneously maximize the size of the reported clauses, while also keeping on eye on reporting ``different'' clauses. One can interpret the negative term as a ``regularizer'' that achieves this goal. 
\begin{definition}[Sketches of binary XOS valuations]
\label{def:bsketch}
For a binary XOS valuation $v$, define its $(k,\alpha)$-sketch $(b_1,...,b_k)$ as $\argmax_{b_1,...,b_k \in \{a_1,...,a_t\}} \sum_{i=1}^m  \left( x_i - \alpha\cdot x_i^2 \right)$, where $a_1,...,a_t$ are the clauses of $v$ and $x_i = \frac{b_1(\{i\}) + \cdots +b_k(\{i\})}{k}$. 
\end{definition}

\begin{remark}
For the sketches defined above, there might be multiple $(b_1,...,b_k)$'s maximize the term. When we use a $(k,\alpha)$-sketch in some protocol, we will use an arbitrary one. Additionally, note that our warmup protocols from Section~\ref{sec:warmup} ask Alice and Bob to output both their $(1,1/2)$-sketch and their $(2, 1/2)$-sketch.
\end{remark}

In Appendix~\ref{app:ske}, we prove some simple properties of these sketches, and an extension of the definition to non-binary XOS valuations. Essentially, what the lemmas are stating is that for any set $A$, the sketches defined above do a ``good enough'' job capturing Alice's (/Bob's) value for $A$. Note that ``good enough'' doesn't mean ``captures $v(A)$ within a constant factor,'' as this is impossible with a concise sketch~\cite{BadanidiyuruDFKNR12}. ``Good enough'' simply means that the sketch can be used inside a similar approach to Section~\ref{sec:warmup}.

Once sketches from Alice and Bob are in hand, there are a couple natural ways to ``wrap up'' the allocation/decision problem. We'll formally name these and refer to them in future protocols:
\begin{itemize}
\item \textbf{Alice-Only Allocation (randomized):} Pick a clause uniformly at random from Alice's sketch, award to Alice items for which that clause values at $1$, and the rest to Bob.
\item \textbf{Best Known Allocation (deterministic):} If Alice reports clauses $a_1,\ldots,a_k$, and Bob reports clauses $b_1,\ldots, b_k$, find $i, j$ maximizing $\SW(a_i, b_j)$. Allocate items according to the allocation that yields $\SW(a_i, b_j)$. 
\item \textbf{Best Known Decision($\alpha, X$) (deterministic):} If Alice reports clauses $a_1,\ldots,a_k$, and Bob reports clauses $b_1,\ldots, b_k$, find $i, j$ maximizing $\SW(a_i, b_j)$. If $\SW(a_i,b_j) \geq \alpha X$ say ``yes'' (guess that $\SW(v_1,v_2) \geq X$). Otherwise, guess ``no'' (guess that $\SW(v_1,v_2) < \alpha X$).
\end{itemize}

\subsection{Our Protocols and Mechanisms}\label{sec:ub}
In this section, we'll describe all protocols used to provide our positive results. All protocols involve Alice and Bob reporting a $(k,\alpha)$-sketch, and then using the Alice-Only or Best Known Allocation, or making the Best Known Decision. All proofs are in Appendix~\ref{app:sbxos} thru~\ref{sec:mechanismsapp}. We make two remarks before proceeding:
\begin{enumerate}
\item All of the high-level intuition for why the protocols work is captured by the sketches. Many of the actual proofs are different, but at a high level everything comes down to the fact that this class of sketches selects ``the right'' clauses to report for welfare maximization.
\item Any protocol that eventually uses the Alice-Only Allocation doesn't require Alice to report her entire sketch (she can just draw the random clause herself as in Protocol~\ref{pro:warmup1} (and have communication $m$ for any choice of $k$). While we state the guarantees for such protocols for a fixed $k$, one can actually take $k \rightarrow \infty$ without increasing the communication at all.
\end{enumerate}

\begin{theorem}\label{thm:main} The following protocols achieve the following guarantees:\\
\begin{tabular}{|c|c|c|c|c|c|}
\hline
Alice's sketch & Bob's Sketch & Wrap-up & Approximation  & Problem&Valuations\\
\hline
$(k,1/2)$ & $\bot$ & Alice-Only & $3/4-1/k$ & Allocation&BXOS\\
\hline
$(k,1/3)$ & $(k,1/3)$ & Best Known Allocation & $23/32-1/k$ & Allocation&XOS\\
\hline
$(k,1/3)$ & $(k,1/3)$ & Best Known Decision & $23/32-1/k$ & Decision&XOS\\
\hline

\end{tabular}
\end{theorem}

Before continuing, we briefly remark the following:
\begin{itemize}
\item The $3/4$-approximation guaranteed by the protocol in the first row is tight: randomized, interactive protocols require exponential communication to beat a $3/4$-approximation. 
\item The second and third protocols also work for general XOS (subject to updating the sketching definition as in Appendix~\ref{app:ske}). 
\item It is still open whether it is possible to beat $23/32$ with a deterministic protocol for the allocation problem, but $23/32$ is optimal for any protocol using the Best Known Allocation after Alice and Bob each report a $(k,\alpha_i)$ sketch (see Lemma~\ref{lem:xoslb} in Appendix~\ref{sec:xos}).
\end{itemize}

Additional applications of our sketches appear in Appendix~\ref{sec:xos2r} and \ref{sec:mechanismsapp}, including a 2-round protocol guaranteeing a $3/4$-approximation for general XOS valuations, and our strictly truthful mechanism. The strictly truthful mechanism essentially visits bidders one at a time, asks for a $(k,1/2)$-sketch on the remaining items, awards them the ``Alice-Only Allocation'' for their reported sketch, and charges payments to ensure strict truthfulness.

\vspace{-1mm}
\section{Lower Bounds}
\label{sec:lb}
\vspace{-2mm}
Finally, we overview our lower bound for the BXOS decision problem (which implies Theorem~\ref{thm:mainlower}). We begin with some intuition: Alice and Bob will each get exponentially many clauses of size $m/2$. These sets will be random, but not uniformly random.\footnote{If they were uniformly random, then Alice and Bob can guarantee $3m/4$ in expectation by just reporting a single arbitrary clause, because two uniformly random sets of size $m/2$ have union $3m/4$ in expectation.} Instead, they are drawn in such a way that the union of two random clauses of Alice and Bob has size $(3/4-1/108)m$ in expectation. At this point, the optimal welfare is $(3/4-1/108)m$ if we don't further adjust their inputs. Finally, we modify the construction either by hiding or not hiding $a_0$ within Alice's input and $b_0$ within Bob's input such that $a_0 \cup b_0 = [m]$, in a matter so that these sets are indistinguishable from the rest. Therefore, the answer to the decision problem rests on whether or not Alice and Bob each have this hidden set, but they have no means by which to convey this information as this set looks indistinguishable from the rest. This description captures all of the intuition for our construction, which appears in Appendix~\ref{sec:lbapp} along with a proof of Theorem~\ref{thm:lb} below.


\begin{theorem}
\label{thm:lb}
For any constant $\varepsilon >0$, there exists a distribution over binary XOS valuations such that no simultaneous, randomized protocol with less than $ e^{2Cm/9}$ communication can guarantee an $\alpha$-approximation to the 2-party BXOS decision problem with probability larger than $\frac{1}{2} +  2e^{-Cm/9}$. Here $\alpha=3/4 - 1/108 + \varepsilon$ and  $C = 2\varepsilon^2$.
\end{theorem}


\section{Discussion and Future Work}
\label{sec:future}
Our main result shows a simultaneous protocol guaranteeing a $3/4$-approximation for the BXOS allocation problem, and a lower bound of $3/4-1/108$ for for the BXOS decision problem. The bigger picture behind these results, even without consideration of truthful combinatorial auctions, is the following:
\begin{itemize}
\item It is surprising that the decision problem is strictly \emph{harder} than the allocation/search problem. To the best of our knowledge, this is the first instance of such a separation.
\item It is surprising that a $(>1/2)$-approximation for either the allocation or decision problem is possible at all, given the strong lower bounds already known on sketching valuation functions, but we are able to get a tight $3/4$-approximation for the allocation problem.
\item A $3/4$-approximation for the decision problem now serves as a new example of what can be achieved in polynomial interactive communication (in fact, two rounds by Theorem~\ref{thm:2rxos2}), but requires exponential simultaneous communication. While such problems are already known, this has a very different flavor than previous constructions, and will likely be a useful tool for this reason.
\end{itemize}

The most obvious question is to resolve whether or not there is a $3/4$-approximation for the allocation problem with general XOS functions. If there isn't, this would provide the first separation between truthful and non-truthful protocols with polynomial communication via Dobzinski's reduction~\cite{Dobzinski16b}. Additionally, whether or not our protocol can be de-randomized is an enticing open question: if no matching deterministic protocol can be found (implying a lower bound of $< 3/4$ for deterministic protocols for the allocation problem), this would provide the first separation between truthful and non-truthful deterministic protocols (Dobzinski's reduction preserves determinism). If our protocol can in fact be de-randomized, this would be fascinating, as this protocol would \emph{deterministically} guarantee a $3/4$-approximation without learning the welfare it achieves.\footnote{It somehow seems tempting to conjecture both that our protocol can be de-randomized and that it can't - a random clause of Alice's does well on average with no input from Bob, so to de-randomize we just need Bob to tell us \emph{something} that identifies a clause performing better than average. At the same time it seems extremely unlikely that a deterministic protocol will somehow provide an approximation guarantee better than $3/4-1/108$ for the allocation problem without violating Theorem~\ref{thm:lb}.}

Finally, while we have provided simultaneous protocols for the allocation problem with approximation guarantees strictly better than $1/2$ when bidders have XOS valuations, it still remains open whether or not a truthful mechanism can obtain a $(>1/2)$-approximation for two-player combinatorial auctions with XOS bidders.

%
\newpage
\thispagestyle{empty}

\appendix







\section{Background on Related Work}\label{sec:related}

There is an enormous literature of related work on combinatorial auctions. The state-of-the-art without concern for incentives is a $1/2$-approximation for any number of subadditive bidders~\cite{Feige06}, and numerous improvements for special cases, such as submodular bidders~\cite{DobzinskiS06,Feige06,FeigeV06}. With concern for incentives, the state-of-the-art (for worst-case approximation ratios and dominant strategy truthfulness) is an $1/O(\sqrt{\log m})$-approximation for XOS bidders, again with improvements for further special cases~\cite{DughmiRY11}. The problem has also been studied in Bayesian settings, where a generic black-box reduction is known if the designer only desires \emph{Bayesian} truthfulness\footnote{A mechanism is Bayesian truthful if it is in every bidder's interest to tell the truth, assuming all other bidders tell the truth and have values drawn from the correct Bayesian prior.}~\cite{HartlineL10, HartlineKM11, BeiH11}. If the designer desires dominant strategy truthfulness but is okay with an average-case welfare guarantee, then a 1/2-approximation is known for XOS bidders~\cite{FeldmanGL14}. Combinatorial auctions have also been studied through the lens of \emph{Price of Anarchy}, but a deeper discussion of this is outside the scope of this paper~\cite{BhawalkarR11, PaesLemeST12, SyrgkanisT12, SyrgkanisT13, FeldmanFGL13, CaiP14, DobzinskiFK15, DevanurMSW15, LehmannLN01, ChristodoulouKS08, BravermanMW16a, FeldmanFMR16}. 

On the topic of simultaneous versus interactive communication, \cite{Yao:1979:CQR:800135.804414} proposed the 2-party simultaneous communication model when communication complexity was introduced.~\cite{Papadimitriou:1982:CC:800070.802192}, \cite{Duris:1984:LBC:800057.808668}, \cite{Nisan:1993:RCC:152322.152386} showed that in the 2-party case, there is an exponential gap between $k$ and $(k-1)$-round deterministic/randomized communication complexity of an explicit function. In the multiparty number-on-forehead communication model \cite{ChandraFL83b}, \cite{BabaiGKL03} showed an exponential gap between simultaneous communication complexity and communication complexity for up to $(\log n)^{1-\varepsilon}$ players for any $\varepsilon >0$. \cite{DobzinskiNO14} recently showed that in combinatorial auctions with unit demand bidders/subadditive bidders, there is an exponential gap (exponential in the number of players) between simultaneous communication complexity and communication complexity. In comparison to these works, our separation between simultaneous and interactive communication for the 2-player BXOS decision problem is of a quite different flavor, and makes the available toolkit for future results more diverse.

\section{Formal Preliminaries}\label{app:formalprelim}
Here, we provide more formal details surrounding our model of communication, approximation, truthful mechanisms, etc.

In a combinatorial auction, there are $n$ players and $m$ items. In 2-party case, we call the first player Alice and the second player Bob. Each player $i$ has a valuation function $v_i:2^{[m]}\rightarrow \R^+$. (We require $v_i(\emptyset) = 0$.) The goal for the auctioneer is to find an allocation $S_1,...,S_n$ ($S_1 \cap \cdots \cap S_n = \emptyset$) to maximize the social welfare $\sum_{i=1}^n v_i(S_i)$. 

\begin{enumerate}
\item When we use ``\textbf{protocol}", it means that players honestly follow the protocol and the challenge is to make the protocol have good approximation ratio, polynomial communication cost and possibly small number of rounds. In this paper, we use the standard communication complexity model and we allow public randomness and private randomness. For details, we refer the reader to \cite{Kushilevitz:1996:CC:264772}.  We want to emphasize two relevant properties of the communication protocols here:
\begin{enumerate}
\item We care about the number of rounds of a protocol. In each round, all the messages need to be sent simultaneously. We use "simultaneous protocols" to denote protocols with only one round of communication. 
\item All the protocols discussed in this paper are in the ``blackboard model.'' In the blackboard model, each message is broadcasted. Or in other words, each message is written on a blackboard for all players and the auctioneer to see. In some protocol, we don't really need broadcast, we will specify where is the message from and sent to in those protocols.
\end{enumerate}

\item When we use ``\textbf{mechanism},'' it means that players might not tell the truth and we need to incentivize the players to cooperate. A mechanism in this paper can be considered as a protocol together with an allocation rule and a payment rule. Let the protocol be $\pi$ and the transcript be $\Pi$. For $i=1,...,n$, let $S_i$ be the allocation rule and $p_i$ be the payment of player $i$. Player $i$'s utility is defined as $u_i(\Pi)  = v_i(S_i(\Pi)) - p_i(\Pi)$. Player $i$'s goal is to maximize her expected utility $\E [u_i(\Pi)] $. The expectation is over the randomness of the mechanism. 
\end{enumerate}

We further define the truthful mechanism as the following. Let $m_i$ be the message sent by player $i$. $m_i$ is a function of $v_i$ and the history of the protocol. Here we only make the definition for the case when each player sends at most one message in the protocol and all the mechanisms in this paper are in this case. We say that $m_i$ is a \textbf{dominant strategy} (in expectation) for player $i$, if for all $v_1,...,v_n$, player $i$'s other strategy $m'_i$ and other players' strategy $m_{-i}$, 
\[
\E \left[ v_i(S_i(\Pi(m_i, m_{-i})) - p_i(\Pi(m_i, m_{-i})) \right] \geq \E \left[ v_i(S_i(\Pi(m'_i, m_{-i})) - p_i(\Pi(m'_i, m_{-i})) \right].
\]
We say that a mechanism is a \textbf{truthful mechanism} if there exist dominant strategies for all players.

One of our goals in this paper is to find an allocation that achieves good approximation of the maximum social welfare $\SW(v_1,...,v_n)$ (defined as allocation problem in Section~\ref{sec:intro}). We say a protocol is $\alpha$-approximation if for all $v_1,...,v_n$, 
\[
\E[  \sum_{i=1}^n v_i(S_i(\Pi))] \geq \alpha \cdot \SW(v_1,...,v_n).
\]
We say a truthful mechanism is $\alpha$-approximation if for all $v_1,...,v_n$ there exist dominant strategies $m_1,...,m_n$ for player $1,...,n$ guaranteeing:
\[
\E[  \sum_{i=1}^n v_i(S_i(\Pi(m_1,...,m_n))] \geq \alpha \cdot \SW(v_1,...,v_n).
\]


\section{Tools for proofs}
\label{sec:prelimapp}
\subsection{Information Theory}

Here we briefly review some facts and definitions from information theory that will be used in this paper. For a more detailed introduction, we refer the reader to \cite{Cover:2006:EIT:1146355}.

Throughout this paper, we use $\log$ to refer to the base $2$ logarithm and use $\ln$ to refer to the natural logarithm. 

\begin{definition}
The \emph{entropy} of a random variable $X$, denoted by $H(X)$, is defined as $H(X) = \sum_x \Pr[X = x] \log(1 / \Pr[X = x])$. 
\end{definition}

If $X$ is drawn from Bernoulli distributions $\Ber(p)$, we use $H(p) = -(p\log p + (1-p)(\log(1-p))$ to denote $H(X)$. 

\begin{definition}
The \emph{conditional entropy} of random variable $X$ conditioned on random variable $Y$ is defined as $H(X|Y) = \mathbb{E}_y[H(X|Y = y)]$. 
\end{definition}

\begin{fact}
$H(XY) = H(X) + H(Y|X)$. 
\end{fact}

\begin{definition}
\label{def:muinfo}
The \emph{mutual information} between two random variables $X$ and $Y$ is defined as $I(X;Y) = H(X) - H(X|Y) = H(Y) - H(Y|X)$. 
\end{definition}

\begin{definition}
The \emph{conditional mutual information} between $X$ and $Y$ given $Z$ is defined as $I(X;Y|Z) = H(X|Z) - H(X|YZ) = H(Y|Z) - H(Y|XZ)$. 
\end{definition}

\begin{fact}\label{fact:cr}
Let $X_1,X_2,Y,Z$ be random variables, we have $I(X_1X_2;Y|Z) = I(X_1;Y|Z) + I(X_2;Y|X_1Z)$.
\end{fact}

\begin{fact}
\label{fact:it1}
Let $X,Y,Z,W$ be random variables. If $I(Y;W|X,Z) = 0$, then $I(X;Y|Z) \geq I(X;Y|ZW)$. 
\end{fact}

\begin{fact}
\label{fact:it2}
Let $X,Y,Z,W$ be random variables. If $I(Y;W|Z) = 0$, then $I(X;Y|Z) \leq I(X;Y|ZW)$. 
\end{fact}

\subsection{Concentration Bound}

\begin{definition}[Negative Correlation]
Let $X_1,...,X_n$ be $n$ random variables supported on $\{0,1\}$. We say $X_1,...,X_n$ are negatively correlated if for all $S \subseteq [n]$,
\[
\Pr\left[\bigwedge_{i \in S} (X_i = 1) \right]  \leq \prod_{i \in S} Pr[X_i = 1].
\]
\end{definition}

\begin{lemma}[Generalized Chernoff Bound\cite{doi:10.1137/S0097539793250767}]
Let $X_1,...,X_n$ be $n$ random variables supported on $\{0,1\}$ and they are negatively correlated. Then for any $a > 0$, 
\[
\Pr\left[ \sum_{i=1}^n X_i \geq a + \E[\sum_{i=1}^n X_i ]\right] \leq e^{-2a^2 /n}. 
\]
\end{lemma}


\section{Missing Definitions and Proofs of Section \ref{sec:ske}}\label{app:ske}

Here we give an extension of Definition \ref{def:bsketch} to non-binary XOS valuations. It is easy to verify that Definition~\ref{def:bsketch} and Definition~\ref{def:sketch} are equivalent for binary XOS valuations.
\begin{definition}[Sketches of XOS valuations]
\label{def:sketch}
For a XOS valuation $v$, define its $(k,\alpha)$-sketch $(b_1,...,b_k)$ as
\[
\argmax_{b_1,...,b_k \in \{a_1,...,a_t\}} \sum_{i=1}^m  \left(  \int_0^{+\infty}(x_{i,u} - \alpha \cdot x_{i,u}^2)du\right).
\]
Here $a_1,...,a_t$ are the clauses of $v$ and $x_{i,u} = \frac{\sum_{j=1}^k \1_{b_j(\{i\}) \geq u}}{k}$. 
\end{definition}

Now we prove some simple properties of these sketches.
\begin{lemma}
\label{lem:bsketch}
Let $(b_1,...,b_k)$ be the $(k,\alpha)$-sketch of some binary XOS valuation $v$. Let $\alpha \leq 1/2$. Let $a$ be a clause of $v$ and $A \subseteq \{i|a(\{i\} =1\}$. We have
\[
\sum_{i=1}^m (x_i -2\alpha\cdot x_i^2) + 2\alpha \cdot  \sum_{i \in A} x_i \geq |A|-\frac{2\alpha \cdot v([m])}{k}.
\]
Here $x_i = \frac{b_1(\{i\}) + \cdots +b_k(\{i\})}{k}$. 
\end{lemma}
\begin{proof}
Because $x_i \leq 1$ and $\alpha \leq 1/2$, we have 
\[
2\alpha \cdot  \sum_{i \in  \{i|a(\{i\} =1\} \backslash A} x_i \leq | \{i|a(\{i\} =1\} \backslash A|.
\]
 So if we can prove the lemma for the case when $A = \{i|a(\{i\} =1\}$, it will directly imply the lemma for the case when $A \subsetneq \{i|a(\{i\} =1\}$. From now on, we will assume $A = \{i|a(\{i\} =1\}$.  By Definition \ref{def:bsketch}, $(b_1,...,b_k)$ are some clauses of $v$ that maximize $\sum_{i=1}^m  (x_i - \alpha \cdot x_i^2)$. 

For $1 \leq j \leq m$, if we replace $b_j$ with $a$, then $\sum_{i=1}^m  (x_i - \alpha \cdot x_i^2)$ will not increase. So we have
\[
\sum_{i=1}^m  (x_i - \alpha \cdot x_i^2) \geq \sum_{i=1}^m \left( \left(x_i - b_j(\{i\})/k + a(\{i\})/k \right)-\alpha \cdot \left(x_i - b_j(\{i\})/k + a(\{i\}/k) \right)^2\right).
\]
This implies
\[
\sum_{i=1}^m \left( \frac{b_j(\{i\}) - a(\{i\})}{k} + 2\alpha\cdot\frac{k \cdot x_i(a(\{i\})-b_j(\{i\})) + \alpha (b_j(\{i\}) - a(\{i\}))^2}{k^2}\right) \geq 0.
\]
Summing over all $j \in [k]$, we get
\begin{eqnarray*}
\sum_{i=1}^m (x_i -2\alpha\cdot x_i^2) + 2\alpha\cdot\sum_{i \in A} x_i &\geq& |A| - \sum_{i=1}^m \sum_{j=1}^k \alpha\cdot\frac{(b_j(\{i\}) - a(\{i\}))^2}{k^2} \\
&\geq& |A| - \sum_{i=1}^m \sum_{j=1}^k \alpha\cdot\frac{ b_j(\{i\}) + a(\{i\}) }{k^2}\\
&=&|A|-\frac{\alpha \cdot (|A| +  \sum_{i=1}^m x_i)}{k} \\
 &\geq& |A|-\frac{2\alpha \cdot v([m])}{k}.
\end{eqnarray*}
\end{proof}

\begin{lemma}
\label{lem:sketch}
Let $(b_1,...,b_k)$ be the $(k,\alpha)$-sketch of some  XOS valuation $v$. Let $\alpha \leq 1/2$. Let $a$ be a clause of $v$. $a'$ is an additive valuation and for all $i$, $a(\{i\}) \geq a'(\{i\})$. We have
\[
\sum_{i=1}^m  \left(  \int_0^{+\infty}(x_{i,u} - 2\alpha \cdot x_{i,u}^2)du+ 2\alpha \cdot \int_0^{a'(\{i\})}x_{i,u} du \right)  \geq a'([m])-\frac{2\alpha \cdot v([m])}{k}.
\]
Here $x_{i,u} = \frac{\sum_{j=1}^k \1_{b_j(\{i\}) \geq u}}{k}$. 
\end{lemma}

\begin{proof}
Because $x_{i,u} \leq 1$ and $\alpha \leq 1/2$, similarly as Lemma \ref{lem:bsketch}, it suffices to prove the case when $a' = a$.  By Definition \ref{def:sketch}, for $1 \leq j \leq m$, if we replace $b_j$ with $a$, $\sum_{i=1}^m  \left(  \int_0^{+\infty}(x_{i,u} - \alpha  \cdot x_{i,u}^2)du\right)$ will not decrease. Let $y_{i,j,u} = (\1_{b_j(\{i\}) \geq u}- \1_{a(\{i\}) \geq u})/k$. We get
\[
\sum_{i=1}^m  \left(  \int_0^{+\infty}(x_{i,u} - \alpha  \cdot x_{i,u}^2)du\right) \geq \sum_{i=1}^m  \left(  \int_0^{+\infty}\left(x_{i,u}-y_{i,j,u}  - \alpha  \cdot (x_{i,u}-y_{i,j,u})^2\right)du\right)
\]
This implies
\begin{eqnarray*}
\sum_{i=1}^m \int_0^{+\infty}\left( y_{i,j,u}- 2\alpha\cdot x_{i,u} \cdot y_{i,j,u} \right) d_u &\geq& -\alpha \cdot \sum_{i=1}^m \int_0^{+\infty} y_{i,j,u}^2 d_u\\
 &\geq& -\alpha \cdot \sum_{i=1}^m \int_0^{\max(a(\{i\}),b_j(\{i\}))} \frac{1}{k^2} \cdot  d_u \\
 &=& -\alpha \cdot\sum_{i=1}^m  \frac{\max(a(\{i\}),b_j(\{i\}))}{k^2} \geq -\frac{2\alpha \cdot v([m])}{k^2}.
\end{eqnarray*}
Notice that $\sum_{j=1}^k  y_{i,j,u} = x_{i,u} -\1_{a(\{i\})\geq u}$. Summing over all $j \in [m]$, we get
\begin{eqnarray*}
\sum_{i=1}^m \int_0^{+\infty}\left( x_{i,u} - 2\alpha\cdot x_{i,u}^2 \right) d_u  &\geq& \sum_{i=1}^m \int_0^{+\infty}\left( \1_{a(\{i\})\geq u}- 2\alpha\cdot x_{i,u} \cdot \1_{a(\{i\})\geq u} \right) d_u - \frac{2\alpha \cdot v([m])}{k} \\
&=& a([m])-\frac{2\alpha \cdot v([m])}{k} - \sum_{i=1}^m 2\alpha \cdot \int_0^{a(\{i\})}x_{i,u} du .
\end{eqnarray*}
\end{proof}



\section{Protocols for BXOS parties}
\label{app:sbxos}

\subsection{Simultaneous protocol for two BXOS parties (implies Theorem~\ref{thm:mainupper})}
\label{sec:bxos}
\begin{algorithm}
	\floatname{algorithm}{Protocol}
        \caption{Simultaneous protocol for 2-party combinatorial auctions with binary XOS valuations}
        \label{pro:bi2}
    \begin{algorithmic}[1]
	\STATE Alice computes the $(k,1/2)$-sketch $(b_1,...,b_k)$ of her valuation $v_1$. Then she picks $j$ uniformly randomly from $\{1,...,k\}$ and sends $b_j$ to the auctioneer.
	\STATE For each item $i$, the auctioneer allocates it to Alice if $b_j(\{i\}) =1$; otherwise allocate it to Bob. 
             \end{algorithmic}
\end{algorithm}

\begin{theorem}
\label{thm:simbxos}
Protocol \ref{pro:bi2} gives a $(3/4-1/k)$-approximation in expectation to the 2-party BXOS allocation problem. 
\end{theorem}


\begin{proof}
Let $S$ and $T = [m] \backslash S$ be the allocation that achieves the optimal social welfare between $v_1$ and $v_2$ (if there are multiple such allocations, just pick an arbitrary one). Let $a$ be some clause of $v_1$ such that $v_1(S) = a(S)$ and $a'$ be some clause of $v_2$ such that $v_2(T) = a'(T)$. Let $A = \{i | a(\{i\} ) = 1\} \cap S$ and  $A' = \{i | a'(\{i\} ) = 1\} \cap T$. We have $A \cap A' =\emptyset$. The optimal social welfare is 
\[
\SW(v_1,v_2) = v_1(S) + v_2(T) = a(S) + a'(T) =  |A| + |A'|.
\]

Let $x_i = \frac{b_1(\{i\}) + \cdots +b_k(\{i\})}{k}$. By Lemma \ref{lem:bsketch}, we have
\[
\sum_{i=1}^m (x_i -x_i^2) +   \sum_{i \in A} x_i \geq |A|-\frac{  v_1([m])}{k}.
\]
Therefore, 
\begin{eqnarray*}
\sum_{i \not \in A'} x_i &\geq& -\sum_{i \in A\cup A'} x_i + \sum_{i=1}^m x_i^2 +|A| -\frac{v_1([m])}{k} \\
&\geq& -\sum_{i \in A\cup A'} x_i + \sum_{i \in A\cup A'} x_i^2 +|A| -\frac{v_1([m])}{k}\\
&=& \sum_{i \in A\cup A'} \left((x_i-\frac{1}{2})^2 - \frac{1}{4}\right)+|A| -\frac{v_1([m])}{k}\\
&\geq& \sum_{i \in A\cup A'} \left( - \frac{1}{4}\right)+|A| -\frac{v_1([m])}{k}\\
&=& -\frac{|A|+|A'|}{4} + |A| -\frac{v_1([m])}{k}.
\end{eqnarray*}

Define $B_j = \{i | b_j(\{i\} ) = 1\}$ and $\overline{B_j} = [m]  \backslash B_j$. By Claim \ref{claim:binaryadditive}, the expected social welfare Protocol \ref{pro:bi2} gets is 
\begin{eqnarray*}
&& \frac{1}{k} \cdot \sum_{j=1}^k \SW(b_j, v_2) \geq \frac{1}{k} \cdot \sum_{j=1}^k \left( b_j(B_j) + a'(\overline{B_j}) \right) \geq \frac{1}{k} \cdot \sum_{j=1}^k \left(  |B_j| + |\overline{B_j} \cap A'| \right) \\
&=& \frac{1}{k} \cdot \sum_{j=1}^k \left(  |B_j| - |A' \cap B_j|+ |A'|  \right) = |A'| +  \frac{1}{k} \cdot \sum_{j=1}^k \sum_{i \not\in A'} b_j(\{i\}) = |A'| + \sum_{ i \not\in A'} x_i \\
&\geq&  |A'| -\frac{|A|+|A'|}{4} + |A| -\frac{v_1([m])}{k} = \frac{3}{4} \cdot (|A| + |A'|) -\frac{v_1([m])}{k}.
\end{eqnarray*}
As $v_1([m]) \leq \SW(v_1,v_2) =  |A| + |A'|$, Protocol \ref{pro:bi2} gives a $(3/4 - 1/k)$ approximation in expectation to the problem.
\end{proof}

\subsection{Sequential protocol for multiple BXOS parties}
\label{sec:bxosm}
\begin{algorithm}[ht]
	\floatname{algorithm}{Protocol}
        \caption{Sequential protocol for multi-party combinatorial auctions with binary XOS valuations}
       \label{pro:bin}
    \begin{algorithmic}[1]
    	\FOR {$l = 1,...,n-1$}
	\STATE The $l$-th player computes the $(k,1/2)$-sketch $(b^l_1,...,b^l_k)$ of her valuation $v_l$ for items that are left. Then she picks $j$ uniformly randomly from $\{1,...,k\}$ and broadcasts $b^l_j$.
	\STATE For each item $i$ left, the auctioneer allocates it to the $l$-th player if $b^l_j(\{i\}) =1$.
       	\ENDFOR
	\STATE Allocate the left items to the $n$-th player.
       \end{algorithmic}

\end{algorithm}
\begin{theorem}
\label{thm:simbxosm}
Protocol \ref{pro:bin} gives a $(1/2-1/k)$-approximation in expectation to the $n$-party BXOS allocation problem. 
\end{theorem}


\begin{proof}
Let $C(n)-1/k$ be the approximation ratio of Protocol \ref{pro:bin} on $n$ players. We prove the theorem by induction on $n$. By Theorem \ref{pro:bi2}, $C(2) \geq 3/4$. For any $n \geq 3$, let $S_1, S_2, ...,S_n$ be the allocation that achieves the optimal social welfare. For $l \in [n]$, let $a_l$ be the clause of $v_l$ such that $a_l(S_l) = v_l(S_l)$ and $A_l = \{i:a_l(\{i\}) = 1\} \cap S_l$. Therefore $\SW(v_1,...,v_n) = |A_1| + \cdots + |A_n|$. 

Let $x_i = \frac{b^1_1(\{i\}) + \cdots +b^1_k(\{i\})}{k}$. By Lemma \ref{lem:bsketch}, we have
\[
\sum_{i=1}^m (x_i - x_i^2) +  \sum_{i \in A_1} x_i \geq |A_1|-\frac{|A_1| + \sum_{i=1}^m x_i}{2k}.
\]
Notice this not exactly from the statement of Lemma \ref{lem:bsketch}, but it is explicit from the proof of Lemma \ref{lem:bsketch}.

 Let $B_j = \{i|b^1_j(\{i\})=1\}$. In Protocol \ref{pro:bin}, player 1 will get welfare $\frac{1}{k} \cdot \sum_{j=1}^k |B_j| = \sum_{i=1}^m x_i$ in expectation. Let $A' = A_2 \cup \cdots \cup A_n$. After player 1 takes all the items in $B_j$, the other players can at least get welfare $|A' \backslash B_j|$ if we allocate items optimally. By induction, in Protocol \ref{pro:bin}, players 2 to $n$ will get welfare $(C(n-1) - 1/k) \cdot |A' \backslash B_j|$ conditioned on player 1 gets $B_j$. So in expectation, players 2 to $n$ will get welfare at least $\frac{1}{k} \cdot \sum_{j=1}^k (C(n-1) - 1/k) \cdot |A' \backslash B_j| = (C(n-1) - 1/k) \cdot(|A'| - \sum_{i \in A'} x_i)$. So in expectation, Protocol \ref{pro:bin} gets welfare at least
\begin{eqnarray*}
 &&\sum_{i=1}^m x_i + (C(n-1) - 1/k) \cdot(|A'| - \sum_{i \in A'} x_i) \\
 &\geq& \sum_{i=1}^m x_i + (C(n-1) - 1/k) \cdot(|A'| - \sum_{i \in A'} x_i) \\
 &&+ C(n-1) \left( |A_1|-\frac{|A_1| + \sum_{i=1}^m x_i}{2k} - \sum_{i=1}^m (x_i - x_i^2) -  \sum_{i \in A_1} x_i \right)\\
&=& C(n-1) \cdot (|A'| + |A_1|) + \sum_{i \in A_1 \cup A'} \left( \left(1-2C(n-1)\right) x_i + C(n-1)x_i^2\right)\\
&& +  \sum_{i \not \in A_1 \cup A'} \left((1-C(n-1)) x_i + C(n-1) x_i^2\right)  - \frac{1}{k} \left(|A'| - \sum_{i \in A'} x_i + C(n-1) \cdot \frac{|A_1| + \sum_{i=1}^m x_i}{2} \right) \\
&\geq&C(n-1) \cdot (|A'| + |A_1|)  +  \sum_{i=1}^m C(n-1) \left(\left(x_i - \frac{2C(n-1)-1}{2C(n-1)}\right)^2 - \frac{(2C(n-1)-1)^2}{4C(n-1)} \right) \\
&& + 0 - \frac{ (|A_1| + |A'|) + (\sum_{i \not \in A'} x_i + |A'|)}{2k}\\
&\geq& \left(C(n-1) -\frac{(2C(n-1)-1)^2}{4C(n-1)} - \frac{1}{k}\right) \cdot \SW(v_1,...,v_n)\\
&=& \left(1 - \frac{1}{4C(n-1)} -\frac{1}{k}\right) \cdot \SW(v_1,...,v_n).
\end{eqnarray*}
So we have
\[
C(n) \geq 1 - \frac{1}{4C(n-1)} = \frac{1}{2} + \frac{C(n-1) - \frac{1}{2}}{2C(n-1)}.
\]
This means $C(n-1) \geq 1/2$ would imply $C(n) \geq 1/2$. As $C(2) \geq 3/4$, we know that $C(n) \geq 1/2$ for all $n \geq 2$. 
\end{proof}



\section{Protocols for XOS parties}
\label{sec:xosapp}

\subsection{Simultaneous protocol for two XOS parties (implies part of Theorem~\ref{thm:mainlower})}\label{sec:xos}
\begin{algorithm}[ht]
	\floatname{algorithm}{Protocol}
        \caption{Deterministic, simultaneous protocol for 2-party combinatorial auctions with  XOS valuations}
        \label{pro:xos2}
    \begin{algorithmic}[1]
	\STATE Alice sends the $(k,1/3)$-sketch $(b^1_1,...,b^1_k)$ of her valuation $v_1$ and Bob sends the  $(k,1/3)$-sketch $(b^2_1,...,b^2_k)$ of his valuation $v_2$ simultaneously to the auctioneer.
	\STATE \textbf{For allocation problem:} Auctioneer finds $j,j'$ maximize $\SW(b^1_j, b^2_{j'})$ and allocates items using the allocation that achieves $\SW(b^1_j, b^2_{j'})$. 
	\STATE \textbf{For decision problem:} Let $X$ be the parameter in the decision problem. Auctioneer finds $j,j'$ maximize $\SW(b^1_j, b^2_{j'})$. If $\SW(b^1_j, b^2_{j'}) \geq (23/32-1/k)X$, say "yes" ($\SW(v_1,v_2) \geq X$). If $\SW(b^1_j, b^2_{j'}) < (23/32-1/k)X$,  say "no". 
             \end{algorithmic}
\end{algorithm}

\begin{theorem}
\label{thm:simxos}
Protocol \ref{pro:xos2} gives a $(23/32 -1/k)$-approximation to the 2-party XOS allocation problem and the 2-party XOS decision problem. 
\end{theorem}


\begin{proof}
Let allocation ($S$, $T = [m] \backslash S$) achieves $\SW(v_1,v_2)$. Let $a_1$ be the clause of $v_1$ such that $a_1(S) = v_1(S)$ and $a_2$ be the clause of $v_2$ such that $a_2(T) = v_2(T)$. Let $a'_1, a'_2$ be two additive valuations such that $a'_1(R) = a_1(S \cap R)$ and $a'_2(R) = a_2(T \cap R)$ for all $R \subseteq [m]$. We have
\[
\SW(v_1,v_2) = v_1(S) + v_2(T) = a_1(S) + a_2(T) = a'_1(S) + a'_2(T). 
\]
Let $x_{i,u} = \frac{\sum_{j=1}^k \1_{b^1_j(\{i\}) \geq u}}{k}$ and $y_{i,u} = \frac{\sum_{j=1}^k \1_{b^2_{j}(\{i\}) \geq u}}{k}$. By Lemma \ref{lem:sketch}, we have
\[
\sum_{i=1}^m  \left(  \int_0^{+\infty}(x_{i,u} - \frac{2}{3} \cdot x_{i,u}^2)du+ \frac{2}{3} \cdot \int_0^{a'_1(\{i\})}x_{i,u} du \right)  \geq a'_1(S)-\frac{\frac{2}{3} \cdot v_1([m])}{k}
\]
and
\[
\sum_{i=1}^m  \left(  \int_0^{+\infty}(y_{i,u} - \frac{2}{3} \cdot  y_{i,u}^2)du+ \frac{2}{3} \cdot \int_0^{a'_2(\{i\})}y_{i,u} du \right)  \geq a'_2(T)-\frac{\frac{2}{3} \cdot v_2([m])}{k}.
\]
Therefore,
\begin{eqnarray*}
&& \frac{3}{4} \cdot \sum_{i=1}^m  \left(  \int_0^{+\infty}(x_{i,u} +y_{i,u}) du \right) \\
&\geq& \frac{3}{4} \left(a'_1(S) + a'_2(T)\right) + \frac{1}{2} \cdot \sum_{i=1}^m  \left(  \int_0^{+\infty}(x^2_{i,u} +y^2_{i,u}) du \right) \\
&& -\frac{1}{2} \cdot  \sum_{i \in S} \left(\int_0^{a'_1(\{i\})}x_{i,u} du\right) - \frac{1}{2} \cdot  \sum_{i \in T} \left(\int_0^{a'_2(\{i\})}y_{i,u} du\right)- \frac{v_1([m]) +  v_2([m])}{2k}.
\end{eqnarray*}
Then we have
\begin{eqnarray*}
&& \frac{1}{k^2} \cdot \sum_{j=1}^k\sum_{j'=1}^k \SW(b^1_j, b^2_{j'}) \\
&=& \frac{1}{k^2} \cdot \sum_{j=1}^k\sum_{j'=1}^k \sum_{i=1}^m \int_0^{+\infty}\left( \1_{b^1_j(\{i\})} +\1_{b^2_{j'}(\{i\})} - \1_{b^1_j(\{i\})} \cdot\1_{b^2_{j'}(\{i\})} \right) du  \\
&=&  \sum_{i=1}^m \int_0^{+\infty}\left( x_{i,u} + y_{i,u} - x_{i,u} \cdot y_{i,u} \right) du  \\
&=&  \frac{3}{4} \cdot \sum_{i=1}^m  \left(  \int_0^{+\infty}(x_{i,u} +y_{i,u}) du \right) +\sum_{i=1}^m \int_0^{+\infty}\left( \frac{1}{4}\cdot( x_{i,u} + y_{i,u}) - x_{i,u} \cdot y_{i,u} \right) du  \\
&\geq&  \frac{3}{4} \left(a'_1(S) + a'_2(T)\right) + \frac{1}{2} \cdot \sum_{i=1}^m  \left(  \int_0^{+\infty}(x^2_{i,u} +y^2_{i,u}) du \right) \\
&& -\frac{1}{2} \cdot  \sum_{i \in S} \left(\int_0^{a'_1(\{i\})}x_{i,u} du\right) - \frac{1}{2} \cdot  \sum_{i \in T} \left(\int_0^{a'_2(\{i\})}y_{i,u} du\right)- \frac{v_1([m]) +  v_2([m])}{2k}\\
&&+\sum_{i=1}^m \int_0^{+\infty}\left( \frac{1}{4}\cdot( x_{i,u} + y_{i,u}) - x_{i,u} \cdot y_{i,u} \right) du\\
&=& \frac{3}{4} \left(a'_1(S) + a'_2(T)\right)  - \frac{v_1([m]) +  v_2([m])}{2k} \\
&& +\sum_{i \in S} \left(\int_0^{a'_1(\{i\})} \left(\frac{1}{2}\left(x_{i,u} - y_{i,u} - \frac{1}{4}\right)^2 - \frac{1}{32}\right)du + \int_{a'_1(\{i\}}^{+\infty}\left(\frac{1}{2}(x_{i,u} - y_{i,u})^2 + \frac{1}{4}(x_{i,u} + y_{i,u}) \right)du\right) \\
&& +\sum_{i \in T} \left(\int_0^{a'_2(\{i\})} \left(\frac{1}{2}\left(y_{i,u} -x_{i,u} - \frac{1}{4}\right)^2 - \frac{1}{32}\right)du + \int_{a'_2(\{i\}}^{+\infty}\left(\frac{1}{2}(x_{i,u} - y_{i,u})^2 + \frac{1}{4}(x_{i,u} + y_{i,u}) \right)du\right) \\
&\geq& \frac{3}{4} \left(a'_1(S) + a'_2(T)\right)  - \frac{v_1([m]) +  v_2([m])}{2k} - \frac{1}{32}  \left(a'_1(S) + a'_2(T)\right) \\
&\geq& (\frac{23}{32} -\frac{1}{k}) \cdot \SW(v_1,v_2).
\end{eqnarray*}
So what Protocol \ref{pro:xos2} gets is at least  
\[
\max_{j,j' \in [k]} \SW(b^1_j, b^2_{j'}) \geq  \frac{1}{k^2} \cdot \sum_{j=1}^k\sum_{j'=1}^k \SW(b^1_j, b^2_{j'}) \geq (\frac{23}{32} -\frac{1}{k}) \cdot \SW(v_1,v_2).
\]
\end{proof}

While we don't know whether or not $23/32$ is in general tight for simultaneous protocols for XOS valuations, the analysis of Protocol~\ref{pro:xos2} (and protocols like it) is indeed tight:
\begin{lemma}
\label{lem:xoslb}
For any constant $\varepsilon > 0$, there exist binary XOS valuations $v_1$ and $v_2$ with sufficiently large $m$ such that simultaneous protocols in the following form won't get $\SW(b^1_j, b^2_{j'})$ larger than $(23/32 + 2\varepsilon)\SW(v_1,v_2)$:
\begin{enumerate}
\item Alice sends the $(k_1,\alpha_1)$-sketch $(b^1_1,...,b^1_{k_1})$ of her valuation $v_1$ and Bob sends the  $(k_2,\alpha_2)$-sketch $(b^2_1,...,b^2_{k_2})$ of his valuation $v_2$ simultaneously to the auctioneer. ($0\leq \alpha_1,\alpha_2 \leq 1/2$, $k_1,k_2$ can be any positive integer)
\item Auctioneer finds $j,j'$ maximize $\SW(b^1_j, b^2_{j'})$ and allocates items using the allocation that achieves $\SW(b^1_j, b^2_{j'})$. 
\end{enumerate}
\end{lemma}

\begin{proof}
Let's assume $\varepsilon < 1/8$. We can do this because the result for smaller $\varepsilon$ will imply the result for larger $\varepsilon$. Let $m$ be a multiple of 4 and we will specify how large $m$ should be later in the proof. Let parameter $\gamma = 5/8 + \sqrt{3} /8$. Let $A_1= \{1,...,m/4\}$, $A_2 = \{m/4+1,...,m/2\}$, $A_3 = \{m/2+1,...,3m/4\}$ and $A_4 = \{3m/4+1,...,m\}$. Let $t =128 / \varepsilon$. We construct $v_1$ and $v_2$ by the following procedure: 
\begin{enumerate}
\item Let $a_0^1$ be the binary additive valuation such that $a_0^1(\{i\}) = 1$ if $i \in A_1\cup A_2$ and $a_0^1(\{i\}) = 0$ if $i \in A_3\cup A_4$. Let $a_0^2$ be the binary additive valuation such that $a_0^2(\{i\}) = 0$ if $i \in A_1\cup A_2$ and $a_0^2(\{i\}) = 1$ if $i \in A_3\cup A_4$.
\item Let $a_1^1,...,a_t^1, a_1^2,...,a_t^2$ be $2t$ binary additive valuations. For $l = 1,...,t$, $j = 1,..4$ and $i \in A_j$,  set $a_l^1(\{i\})$ to be 1 with probability $p_j$ independently. For $l = 1,...,t$, $j = 1,..4$ and $i \in A_j$,  set $a_l^2(\{i\})$ to be 1 with probability $q_j$ independently. Here $p_1 = 5/4-\gamma + \varepsilon$, $p_2 = \gamma +\varepsilon$, $p_3 = 1 - \gamma$, $p_4 = \gamma -1/4$, $q_1 = 1-\gamma$, $q_2 = \gamma -1/4$, $q_3 = 5/4-\gamma +\varepsilon$, $q_4 = \gamma +\varepsilon$. 
\item Let $v_1$ be the binary XOS valuation with $t+1$ clauses $a_0^1,...,a_t^1$. Let $v_2$ be the binary XOS valuation with $t+1$ clauses $a_0^2,...,a_t^2$. 
\end{enumerate}
We first show that with positive probability, the following conditions are all satisfied:
\begin{enumerate}
\item $(1/2 + 7\varepsilon/16)m \leq a_j^1([m])\leq (1/2+\varepsilon)m$, $\forall j \in [t]$.
\item $(1/2 + 7\varepsilon/16)m \leq a_j^2([m]) \leq (1/2+\varepsilon)m$, $\forall j \in [t]$.
\item $(5/16 + 7\varepsilon/16)m \leq a_j^1(A_1 \cup A_2)\leq m/2$, $\forall j \in [t]$.
\item $(5/16 + 7\varepsilon/16)m \leq a_j^2(A_3 \cup A_4) \leq m/2$, $\forall j \in [t]$.
\item $\sum_{i=1}^m a_j^1(\{i\}) a_l^1(\{i\}) \leq (5/16 + 3\varepsilon/4) m$ ,$\forall j,l\in [t], j\neq l$.
\item  $\sum_{i=1}^m a_j^2(\{i\}) a_l^2(\{i\}) \leq (5/16 + 3\varepsilon/4) m$ ,$\forall j,l\in [t], j\neq l$. 
\item $\sum_{i=1}^m a_j^1(\{i\}) a_l^2(\{i\}) \geq 9m /32$, $\forall j,l\in [t]$. 
\end{enumerate}

The mean values of these terms are:
\begin{enumerate}
\item $\E[ a_j^1([m])]  =  \E[a_j^2([m])] = (1/2+\varepsilon/2)m$, $\forall j \in [t]$.
\item $\E[a_j^1(A_1 \cup A_2)]=\E[ a_j^2(A_3 \cup A_4)] = (5/16 + \varepsilon/2)m$, $\forall j \neq [t]$.
\item $\E[\sum_{i=1}^m a_j^1(\{i\}) a_l^1(\{i\}) ] = \E[\sum_{i=1}^m a_j^2(\{i\}) a_l^2(\{i\})] = (5/16 + 5\varepsilon/8 + \varepsilon^2/2) m\leq (5/16 + 11\varepsilon/16)m$, $\forall j,l\in [t], j\neq l$.
\item $\E[\sum_{i=1}^m a_j^1(\{i\}) a_l^2(\{i\})] = (9/32 +3\varepsilon / 8)m$, $\forall j,l\in [t]$. 
\end{enumerate}
By Chernoff bound, each of these conditions is satisfied with probability $1 - e^{-\Omega(\varepsilon^2m)}$. There are $2t^2 + 7t$ conditions. So by union bound, all of them are satisfied with probability $\geq 1 - (2t^2 + 7t) e^{-\Omega(\varepsilon^2m)}$. By picking $m$ large enough, this probability will be positive. So we can find $v_1$ and $v_2$ sampled from our procedure and satisfy all these conditions. 

Now we assume $v_1$ and $v_2$ satisfy all these conditions. We have $\SW(v_1,v_2) = \SW(a_0^1,a_0^2) = m$. On the other hand, $\forall j,l\in [t]$, 
\[
\SW(a_j^1, a_l^2) = \sum_{i=1}^m \left(-a_j^1(\{i\}) a_l^2(\{i\}) + a_j^1(\{i\})+a_l^2(\{i\}) \right) \leq 2\cdot (1/2+\varepsilon)m - 9m/32 \leq (23/32+2\varepsilon)m.
\]
So in order to prove the lemma, it suffices to show that both $a_0^1$ and $a_0^2$ cannot not appear in Alice and Bob's sketches for arbitrary positive integers $k_1,k_2$ and $0\leq \alpha_1,\alpha_2 \leq 1/2$. Wlog, we will only show $a_0^1$ cannot not appear in the  $(k,\alpha)$-sketch of $v_1$ for arbitrary positive integer $k$ and $0\leq \alpha \leq 1/2$. 

Let $(b_1,...,b_k)$ to be a $(k,\alpha)$-sketch of $v_1$. Let $f(b_1,...,b_k)$ be the term that this sketch is trying to maximize by choosing $b_1,...,b_k$ over $a_0^1,...,a_t^1$. Recall that
\begin{eqnarray*}
f(b_1,...,b_k) &=& \sum_{i=1}^m \left(\frac{\sum_{j=1}^k b_j(\{i\})}{k} - \alpha \left(\frac{\sum_{j=1}^k b_j(\{i\})}{k}\right)^2\right) \\
&=& \frac{1}{k} \sum_{i=1}^m \left( \sum_{j=1}^k (1-\alpha/k) b_j(\{i\}) -\alpha \sum_{j' \in [k], j'\neq j} b_j(\{i\})b_{j'}(\{i\})/ k \right).
\end{eqnarray*}
Suppose $a_0^1$ appears in the sketch. Wlog assume $b_1 = a_0^1$. Let $b_1' = a_l^1$ such that $|\{j| 2\leq j\leq k, b_j = a_l^1\}| \leq k/t$. We can always find such $l$ by averaging argument. We have
\begin{eqnarray*}
&&f(b_1',b_2,...,b_k) - f(b_1,...,b_k) \\
&=& \frac{1}{k}\sum_{i=1}^m \left(  (1-\alpha/k) (b_1'(\{i\})-b_1(\{i\})) + 2\alpha \sum_{2\leq j\leq k} b_j(\{i\})(b_1(\{i\})-b_1'(\{i\}))/ k \right)
\end{eqnarray*}
We know that 
\[
\sum_{i=1}^m b_1'(\{i\})-b_1(\{i\}) \geq (1/2 + 7\varepsilon/16)m - m/2 = 7\varepsilon m/16.
\]
We have to do a case analysis for term $\sum_{i=1}^m b_j(\{i\})(b_1(\{i\})-b_1'(\{i\}))$:
\begin{enumerate}
\item If $b_j = a_0^1$, then 
\[
\sum_{i=1}^m b_j(\{i\})(b_1(\{i\})-b_1'(\{i\})) \geq m/2 - m/2 = 0.
\]
\item If $b_j = a_l^1$, then 
\[
\sum_{i=1}^m b_j(\{i\})(b_1(\{i\})-b_1'(\{i\})) \geq -\sum_{i=1}^m a_l^1(\{i\}) = -(1/2+\varepsilon)m.
\]
\item If $b_j \neq a_l^1,a_0^1$, then
\[
\sum_{i=1}^m b_j(\{i\})(b_1(\{i\})-b_1'(\{i\})) \geq (5/16 + 7\varepsilon/16)m- (5/16 + 3\varepsilon/4) m = -5\varepsilon m /16.
\]
\end{enumerate}
If $k = 1$, 
\[
f(b_1') - f(b_1) = (1-\alpha) \sum_{i=1}^m b_1'(\{i\})-b_1(\{i\})  > 0.
\]
Now we assume $k >1$. We have
\begin{eqnarray*}
&&f(b_1',b_2,...,b_k) - f(b_1,...,b_k) \\
&\geq& \frac{m}{k} \left( (1-\alpha/k) (7\varepsilon/16) -2\alpha (5\varepsilon /16) -  2\alpha(1/2+\varepsilon)/t \right) \\
&\geq&  \frac{m}{k} \left( (3/4) (7\varepsilon/16) -(5\varepsilon /16) -  (1/2+\varepsilon)/t \right) \\
&\geq&  \frac{\varepsilon m}{k} (1 /64 - 1/ \varepsilon t) > 0.
\end{eqnarray*}
Now we get a contradiction. So $a_0^1$ cannot be in the sketch. 
\end{proof}

\subsection{Two-round protocol for two XOS parties} \label{sec:xos2r}

\begin{algorithm}
	\floatname{algorithm}{Protocol}
        \caption{Two-round, deterministic protocol for 2-party combinatorial auctions with XOS valuations}
        \label{pro:2rxos2}
    \begin{algorithmic}[1]
	\STATE Alice computes the $(k,1/2)$-sketch $(b_1,...,b_k)$ of her valuation $v_1$ and sends it to Bob. Bob find $j$ that maximize $SW(b_j, v_2)$, and sends $SW(b_j,v_2)$ together with the allocation that achieves $SW(b_j,v_2)$ to the auctioneer.
	\STATE \textbf{For allocation problem:} The auctioneer allocates items according to the allocation provided by Bob.
	\STATE \textbf{For decision problem:} Let $X$ be the parameter in the decision problem. If $SW(b_j, v_2) \geq (3/4-1/k)X$, say "yes" ($\SW(v_1,v_2) \geq X$). If $SW(b_j, v_2) < (3/4-1/k)X$, say "no". 
             \end{algorithmic}
\end{algorithm}

\begin{theorem}
\label{thm:2rxos2}
Protocol \ref{pro:2rxos2} gives a $(3/4-1/k)$-approximation to the 2-party XOS allocation problem and the 2-party XOS decision problem. 
\end{theorem}


\begin{proof}
Let allocation ($S$, $T = [m] \backslash S$) achieves $\SW(v_1,v_2)$. Let $a_1$ be the clause of $v_1$ such that $a_1(S) = v_1(S)$ and $a_2$ be the clause of $v_2$ such that $a_2(T) = v_2(T)$. Let $a'_1, a'_2$ be two additive valuations such that $a'_1(R) = a_1(S \cap R)$ and $a'_2(R) = a_2(T \cap R)$ for all $R \subseteq [m]$. We have
\[
\SW(v_1,v_2) = v_1(S) + v_2(T) = a_1(S) + a_2(T) = a'_1(S) + a'_2(T). 
\]

Let $x_{i,u} = \frac{\sum_{j=1}^k \1_{b_j(\{i\}) \geq u}}{k}$. By Lemma \ref{lem:sketch}, we have
\[
\sum_{i=1}^m  \left(  \int_0^{+\infty}(x_{i,u} - x_{i,u}^2)du+  \int_0^{a'_1(\{i\})}x_{i,u} du \right)  \geq a'_1(S)-\frac{ v_1([m])}{k}.
\]
Then we have
\begin{eqnarray*}
&&\frac{1}{k} \sum_{j=1}^k \SW(b_j, v_2)\\
 &\geq& \frac{1}{k} \sum_{j=1}^k \SW(b_j,a'_2)  \\
&=& a'_2([m]) + \sum_{i=1}^m   \int_{a'_2(\{i\})}^{+\infty}x_{i,u}du \\
&=&  a'_2([m]) + \sum_{i=1}^m   \int_{0}^{+\infty}x_{i,u}du  -\sum_{i \in T}   \int_{0}^{a'_2(\{i\})}x_{i,u}du \\
&\geq & a'_2(T) + a'_1(S) +\sum_{i=1}^m   \int_0^{+\infty}x^2_{i,u} d_u- \sum_{i \in S}   \int_{0}^{a'_1(\{i\})}x_{i,u}du -\sum_{i \in T}   \int_{0}^{a'_2(\{i\})}x_{i,u}du -\frac{ v_1([m])}{k}\\
&=& \SW(v_1,v_2) +\sum_{i \in S}   \int_{0}^{a'_1(\{i\})}\left((x_{i,u}-\frac{1}{2})^2-\frac{1}{4}\right)du +  \sum_{i \in T}   \int_{0}^{a'_2(\{i\})}\left((x_{i,u}-\frac{1}{2})^2-\frac{1}{4}\right)du -\frac{ v_1([m])}{k} \\
&\geq&\SW(v_1,v_2) -\frac{1}{4} (a'_2(T) + a'_1(S)) - \SW(v_1,v_2)/k \\
&=& (\frac{3}{4} - \frac{1}{k}) \SW(v_1,v_2).
\end{eqnarray*}

\end{proof}



\section{Truthful mechanism for multiple binary XOS parties}
\label{sec:mechanismsapp}
In this section, we will show how to transform Protocol \ref{pro:bi2} and Protocol \ref{pro:bin} into truthful mechanisms. Since Protocol \ref{pro:bi2} is a special case of Protocol \ref{pro:bin}, we will only do the transformation for Protocol \ref{pro:bin}. 

\begin{algorithm}[ht]
	\floatname{algorithm}{Mechanism}
        \caption{Truthful mechanism for multi-party combinatorial auctions with binary XOS valuations}
       \label{mech:bin}
    \begin{algorithmic}[1]
    	\FOR {$l = 1,...,n-1$}
	\STATE The $l$-th player broadcasts $k$ binary additive valuations $(b^l_1,...,b^l_k)$. 
	\STATE The auctioneer picks $j$ uniformly randomly from $\{1,...,k\}$ and broadcasts $j$. For each item $i$ left, the auctioneer allocates it to the $l$-th player if $b^l_j(\{i\}) =1$ and charges the $l$-th player $\frac{\sum_{j' =1}^k b^l_{j'}(\{i\})}{2k}$.
       	\ENDFOR
	\STATE Allocate the left items to the $n$-th player.
       \end{algorithmic}
\end{algorithm}

\begin{theorem}
\label{thm:truthful}
In Mechanism \ref{mech:bin}, for any $l \in \{1,...,n-1\}$ and any history before the $l$-player, if $(b^l_1,...,b^l_k)$ maximizes her expected utility, then $(b^l_1,...,b^l_k)$ is a $(k, 1/2)$-sketch of $v_l$ for items left.
\end{theorem}

Note that this is not tight - a simple observation shows how to turn the $(1-1/e)$-approximation of~\cite{Feige06, DobzinskiS06} into a truthful mechanism for binary XOS bidders. One noteworthy difference is that in Mechanism~\ref{mech:bin}, following the intended protocol is \emph{strictly} more profitable than any deviation. In the previous mechanisms, all bidders get zero utility no matter what they do. Mechanism~\ref{mech:bin} is also considerably simpler than the previous mechanisms.


\begin{proof}
Wlog, it suffices to prove the theorem for the first player. 

Let $x_i = \frac{b^1_1(\{i\}) + \cdots +b^1_k(\{i\})}{k}$ and $B_j = \{i | b^1_j \{ i\} = 1\}$. The expected utility of the first player is 
\[
\sum_{j=1}^k v_1(B_j) -\frac{k}{2} \cdot  \sum_{i = 1}^m x_i^2.
\]
Assume $(b^1_1,...,b^1_k)$ maximizes the first player's expected utility. We are going to first show that $b^1_1, ...,b^1_k$ are all clauses of $v_1$. We prove by contradiction. Wlog suppose $b^1_1$ is not a clause of $v_1$. We want to show that $(b^1_1,...,b^1_k)$ does not maximize the first player's expected utility. Let $a_t$ be a clause of $v_1$ such that $v_1(B_j) = a_t(B_j)$. Let $A_t =  \{i |a_t( \{ i\}) = 1\}$. There are two cases:
\begin{enumerate}
\item $B_j \subseteq A_t$: If the first player reports $(a_t,b^1_2,...,b^1_k)$ instead of $(b^1_1,...,b^1_k)$, her expected utility will be increased by 
\[
a_t(A_t) - a_t(B_j) + \frac{k}{2} \cdot  \sum_{i \in A_t \backslash B_j} \left(x_i^2 -(x_i + 1/k)^2\right) = \frac{k}{2} \cdot  \sum_{i \in A_t \backslash B_j} (2/k - 2x_i/ k - 1/k^2). 
\]
For each $i \in A_t \backslash B_j$, $x_i \leq 1 - 1/k$, so $\sum_{i \in A_t \backslash B_j} (2/k - 2x_i/ k - 1/k^2) > 0$. So the first player will get strictly more utility if she reports $(a_t,b^1_2,...,b^1_k)$ instead of $(b^1_1,...,b^1_k)$.
\item $B_j$ is not a subset of $A_t$: Let $A' = B_j \cap A_t$. Let $a'$ be the binary additive valuation such that $a'(\{i\}) =\1_{i \in A'}$ for $i = 1,...,m$. We know that  $a' \neq b^1_1$. If the first player reports $(a',b^1_2...,b^1_k)$ instead of $(b^1_1,...,b^1_k)$, her expected utility will be increased by 
\[
a_t(A') - a_t(B_j) + \frac{k}{2} \cdot  \sum_{i \in B_j \backslash A'} \left(x_i^2 -(x_i - 1/k)^2\right).
\]
We know that $a_t(A') = a_t(B_j)$, so above term is strictly larger than 0. So the first player will get strictly more utility if she reports $(a',b^1_2,...,b^1_k)$ instead of $(b^1_1,...,b^1_k)$.
\end{enumerate}
Now we have shown that $b^1_1, ...,b^1_k$ are all clauses of $v_1$. The expected utility of the first player can be written as
\[
\sum_{j=1}^k b^1_j(B_j) -\frac{k}{2} \cdot  \sum_{i = 1}^m x_i^2 = k \cdot \sum_{i=1}^m \left(x_i - \frac{1}{2} \cdot x_i^2\right).
\]
This is $k$ times what we want to maximize in the $(k,1/2)$-sketch. So when $(b^1_1,...,b^1_k)$ maximizes the first player's expected utility, it's also a $(k,1/2)$-sketch of $v_1$.
\end{proof}



\section{Missing Proof of Section \ref{sec:lb}}
\label{sec:lbapp}

We now provide a little more detail for the construction. Any claims in the first six bullets are straight-forward without any background in information theory. Formalizing the final bullet is the only tricky part, but will probably appear straight-forward to experts in information theory. The real interest lies in the construction itself.

\begin{itemize}
\item First, draw $S$ to be a uniformly random set of size $m/2$. Draw $T$ uniformly random among all sets of size $m/2$ with $S \cap T = m/3$. Alice will know $S$ (but not $T$, only that $S \cap T = m/3$), and Bob will know $T$ (but not $S$).
\item Each of Alice's (exponentially many) clauses $a_i$ are drawn uniformly random among sets such that $|a_i \cap S| = m/3$ and $|a_i \cap \bar{S}| = m/6$. In other words, Alice's sets are all of size $m/2$, and all intersect $S$ more than a uniformly random set. 
\item Each of Bob's (exponentially many) clauses $b_i$ are drawn uniformly random among sets such that $|b_i \cap T| = m/3$ and $|b_i \cap \bar{T}| = m/6$. 
\item At this point, with very high probability, the optimal welfare is $(3/4-1/108)m$. 
\item Now, notice that there exist (many possible) $a^0_0$ such that $|a^0_0 \cap S| = m/3 = |a^0_0 \cap T|$. There also exist (many possible) $a^1_0$ such that $|a^1_0 \cap S| = m/3 = |\bar{a}^1_0\cap T|$. So with probability $1/2$, add the clause $a^0_0$ to both Alice and Bob's input. With probability $1/2$, add the clause $a^1_0$ to Alice's clauses and $b^1_0 = \bar{a}^1_0$ to Bob's.
\item We can prove that in case $0$, we haven't improved the welfare at all, while in case $1$ the welfare has improved to $m$. So in order to obtain better than a $(3/4-1/108)$ approximation for the decision problem, Alice and Bob must figure out whether they are in case 0 or case 1.
\item But doing so requires someone to communicate useful information about either $a_0$ or $b_0$, which is impossible since all sets appear a priori indistinguishable and there are exponentially many of them. 
\end{itemize}

\begin{theorem}[Restatement of Theorem \ref{thm:lb}]
\label{thm:lbapp}
For any constant $\varepsilon >0$, there exists a distribution over binary XOS valuations such that no simultaneous protocol with communication cost less than $ e^{2Cm/9}$ can have an $\alpha$-approximation to the 2-party BXOS decision problem with probability larger than $\frac{1}{2} +  2e^{-Cm/9}$. Here $\alpha=3/4 - 1/108 + \varepsilon$ and  $C = 2\varepsilon^2$.
\end{theorem}

\begin{proof}
This proof uses mutual information (Definition \ref{def:muinfo}) and other information theory tools. We also use generalized Chernoff bound for negatively correlated random variables. We mention these tools in Appendix \ref{sec:prelimapp}. 

To start the proof, we sample Alice and Bob's valuations $v_1$ and $v_2$ from the below procedure. It's basically following the ideas above. Sampling in this specific way will make the later part of the proof easier. Here $M$ is the bit we want to hide in Alice and Bob's inputs. Later we will show $M$ is equal to the output of the decision problem with high probability.
\begin{enumerate}
\item Sample $S$ and $T$ uniformly randomly from all the pairs of sets that satisfy $S,T \subseteq [m]$, $|S|=|T| = m/2$, $|S \cap T| = m/3$. 
\item Sample bit $M$ uniformly randomly from $\{0,1\}$. Let $S^c = [m] \backslash S$ and $T^c = [m] \backslash T$. 
\begin{enumerate}
\item If $M = 1$ then sample $U_1$ uniformly randomly from all the sets that satisfy $S\cap T^c \subseteq U_1$, $|U_1 \cap (S \cap T)| = m/6$, $|U_1\cap (S^c \cap T^c)| = m/6$ and $|U_1 \cap (T \cap S^c)| = 0$. And let $U_2 = [m] \backslash U_1$. 
\item If $M = 0$ then sample $U_1$ uniformly randomly from all the sets that satisfy $S\cap T \subseteq U_1$, $|U_1\cap (S^c \cap T^c)| = m/6$, $|U_1 \cap (S \cap T^c)| = 0$ and $|U_1 \cap (T \cap S^c)| = 0$. And let $U_2 = U_1$. 
\end{enumerate}
\item Let $l = e^{4Cm/9}$. Sample $J_1,J_2$ uniformly randomly from $[l]$. Set $A_{J_1} = U_1$ and $B_{J_2} = U_2$. 
\begin{enumerate}
\item Define $D_S$ to be the uniform distribution over all sets $X$ such that $X \subseteq[m]$, $|X \cap S| = m/3$ and $|X \cap S^c| = m/6$. For $j =1,...,l$ and $j \neq J_1$ sample $A_j$ from $D_S$. 
\item Similarly define $D_T$ to be the uniform distribution over all sets $X$ such that $X \subseteq[m]$, $|X \cap T| = m/3$ and $|X \cap T^c| = m/6$. For $j =1,...,l$ and $j \neq J_2$ sample $B_j$ from $D_T$. 
\end{enumerate}
\item Finally we set $v_1$ and $v_2$ as following:
\begin{enumerate}
\item For $j=1,...,l$, define $a_j$ to be the binary additive valuation such that $a_j(\{i\}) =  \1_{i \in A_j}$ for all $i \in [m]$. And set $v_1$ to be the binary XOS valuation with clauses $a_1,...,a_l$. 
\item For $j=1,...,l$, define $b_j$ to be the binary additive valuation such that $b_j(\{i\}) =  \1_{i \in B_j}$ for all $i \in [m]$. And set $v_2$ to be the binary XOS valuation with clauses $b_1,...,b_l$. 
\end{enumerate}
\end{enumerate}
Since we are working on a specific distribution, wlog we can assume the simultaneous protocol is deterministic. Let $\Pi_1$ be the message sent by Alice and $\Pi_2$ be the message sent by Bob in the simultaneous protocol. The rest of the proof proceeds in two steps. We will first show that $\Pi_1$ and $\Pi_2$ together convey very little information about $M$ to the auctioneer using information theoretic argument. After that, we will show that if the auctioneer has very little information about $M$, she cannot solve BXOS decision problem with large probability. 

For the first step, we prove the following lemma. It's basically stating that $I(\Pi_1\Pi_2 ; M)$ - the mutual information (Definition \ref{def:muinfo}) between $\Pi_1\Pi_2$ (messages sent by Alice and Bob) and the bit $M$ is no more than the length of messages over the number of sets Alice/Bob has in its input. Here's an overview of the proof. The proof can be broken down into two parts. In the first part we show $I(\Pi_1\Pi_2 ; M)$ is upper bounded by $I(\Pi_1;U_1|S)+ I(\Pi_2;U_2|T)$. As defined above, $U_1$ and $U_2$ are special sets which contain information about $M$. The proof of the first part is roughly just to use independence of random variables and the fact that $M$ is a function of $U_1$ and $U_2$ to transit through mutual information terms. In the second part of the proof, we show $I(\Pi_1;U_1|S) \leq \frac{|\Pi_1|}{l}$. The idea is that when only $S$ is given, $U_1$ looks similar to other sets in Alice's input. We can prove that the amount of information $\Pi_1$ conveys about $U_1$ is the same as the amount of information $\Pi_1$ conveys about some other set in Alice's input. As $\Pi_1$'s entropy is at most its length, we can conclude $I(\Pi_1;U_1|S) \leq \frac{|\Pi_1|}{l}$.

\begin{lemma}
\label{lem:lbit}
\[
I(\Pi_1\Pi_2 ; M) \leq \frac{|\Pi_1| + |\Pi_2|}{l} =  e^{-2Cm/9}.
\]
\end{lemma}
\begin{proof}
Since $M$ is independent with $S$ and $T$, we have
\[
I(\Pi_1\Pi_2;M) \leq I(\Pi_1\Pi_2ST;M) = I(ST;M) +I(\Pi_1\Pi_2;M|ST) =  0 + I(\Pi_1\Pi_2;M|ST). 
\]
Since $M$ is a function of $U_1$ and $U_2$ ($M = \1_{U_1=U_2}$), we have
\[
I(\Pi_1\Pi_2;M|ST) \leq I(\Pi_1\Pi_2;U_1U_2|ST) = I(\Pi_1;U_1U_2|ST)+I(\Pi_2;U_1U_2|ST\Pi_1).
\]
Notice that 
\begin{eqnarray*}
I(\Pi_2;U_1U_2|ST\Pi_1) &=& I(\Pi_2;U_1U_2\Pi_1|ST) - I(\Pi_2;\Pi_1|ST) \\
&\leq& I(\Pi_2;U_1U_2\Pi_1|ST)\\
&=& I(\Pi_2;U_1U_2|ST) + I(\Pi_2;\Pi_1|STU_1U_2) \\
&\leq& I(\Pi_2;U_1U_2|ST) + I(v_2;v_1|STU_1U_2) \\
&=&  I(\Pi_2;U_1U_2|ST).
\end{eqnarray*}
The second last step is because $\Pi_1$ is a function of $v_1$ and $\Pi_2$ is a function of $v_2$. The last step is because after $S,T,U_1,U_2$ are sampled, $v_1$ and $v_2$ are sampled using independent randomness. 

It's easy to check that when $S,T$ are given, $U_2$ is a function of $U_1$. So we have
\begin{eqnarray*}
I(\Pi_1;U_1U_2|ST) &=& I(\Pi_1;U_1|ST) \\
&=& I(\Pi_1;U_1T|S) - I(\Pi_1;T|SU_1) \\
&\leq& I(\Pi_1;U_1T|S) \\
&=& I(\Pi_1;U_1|S) + I(\Pi_1;T|SU_1) \\
&\leq&I(\Pi_1;U_1|S) + I(v_1;T|SU_1) \\
&=&I(\Pi_1;U_1|S).
\end{eqnarray*}
The second last step is because $\Pi_1$ is a function of $v_1$. The last step is because when fixing $S$ and $U_1$, $v_1$ is independent with $T$. Similarly we also get 
\[
I(\Pi_2;U_1U_2|ST) \leq I(\Pi_2;U_2|T).
\]
So far, we get 
\begin{eqnarray*}
I(\Pi_1\Pi_2;M) &\leq& I(\Pi_1\Pi_2;M|ST) \leq I(\Pi_1;U_1U_2|ST) +I(\Pi_2;U_1U_2|ST\Pi_1) \\
 &\leq& I(\Pi_1;U_1U_2|ST) +I(\Pi_2;U_1U_2|ST) \leq I(\Pi_1;U_1|S)+ I(\Pi_2;U_2|T).
\end{eqnarray*}
Finally, we just need to bound $I(\Pi_1;U_1|S)$ and  $I(\Pi_2;U_2|T)$. First it's easy to check that when $S$ is fixed, $A_{J_1}$ is also distributed as $D_S$. Therefore given $S$, fixing $J_1$ does not change $v_1$'s distribution . So $I(v_1;J_1|S) = 0$. As $\Pi_1$ is a function of $v_1$, we also have $I(\Pi_1;J_1|S) = 0$. Therefore,
\[
I(\Pi_1;U_1|S) = I(\Pi_1;A_{J_1}|S) \leq I(\Pi_1;A_{J_1} J_1|S) = I(\Pi_1;J_1|S) + I(\Pi_1;A_{J_1}|SJ_1) = I(\Pi_1;A_{J_1}|SJ_1). 
\]
We also have
\[
I(\Pi_1; A_1A_2...A_l|S) \leq H(\Pi_1|S) \leq |\Pi_1|. 
\]
On the other hand, we know that for all $j=1,...,l-1$, $I(A_1...A_j;A_{j+1}|S) = 0$. By Fact \ref{fact:it2}, we have
\[
I(\Pi_1;A_1A_2...A_l|S) = \sum_{j=1}^l I( \Pi_1;A_j| SA_1...A_{j-1}) \geq \sum_{j=1}^l I(\Pi_1; A_j|S).
\]
Then we have
\[
I(\Pi_1;U_1|S) \leq I(\Pi_1;A_{J_1}|SJ_1)= \frac{1}{l} \cdot \sum_{j=1}^l I(\Pi_1; A_j|S) \leq \frac{|\Pi_1|}{l}. 
\]
Similarly we have
\[
 I(\Pi_2;U_2|T) \leq \frac{|\Pi_2|}{l}. 
\]
To sum up, we get
\[
I(\Pi_1\Pi_2;M) \leq I(\Pi_1;U_1|S)+ I(\Pi_2;U_2|T) \leq  \frac{|\Pi_1| + |\Pi_2|}{l}. 
\]
\end{proof}
Now we are going to show that the auctioneer cannot solve the decision problem with large probability based on Lemma \ref{lem:lbit}. Let $D$ be the output from the auctioneer for the BXOS decision problem. ($D=1$ if the answer is "yes", $D=0$ if the answer is no.) Since $D$ is a function of $\Pi_1$ and $\Pi_2$, we have
\[
I(D;M) \leq I(\Pi_1\Pi_2;M) \leq  e^{-2Cm/9}.
\]
Let $e = \Pr[D \neq M]$. By Fano's inequality ,
\[
 1- (e-1/2)^2\geq H(e) \geq H(M|D) = H(M) - I(M;D) \geq 1 - e^{-2Cm/9}. 
\]
So $e \geq 1/2 - e^{-Cm/9}$. This basically states that if mutual information between $M$ and $D$ is small, then the probability of $M\neq D$ is large.

Finally we will show that since $\Pr[D \neq M]$ is large, the algorithm fails to solve the decision problem with large probability. Let $W$ be the indicator variable such that $W =1$ if the one of the following three events happens and $W=0$ otherwise:
\begin{enumerate}
\item $\alpha  m \leq \SW(v_1,v_2) < m$. 
\item $\SW(v_1,v_2) < \alpha m$ and $M = 1$.
\item $\SW(v_1,v_2) \geq  m$ and $M = 0$. 
\end{enumerate}
We know that if $W = 0$, then $M$ is the unique correct output of the decision problem. So the probability that auctioneer fails to solve the decision problem is at least $\Pr[M \neq D] - \Pr[W = 1]$.

Now we need to upper bound $Pr[W=1]$. 
\begin{enumerate}
\item When $M = 1$, we know that $\SW(v_1,v_2) \geq \SW(a_{J_1},b_{J_2}) = |U_1 \cup U_2| = m$. So $\Pr[W = 1|M = 1] =0$. 
\item When $M = 0$, we know that $\SW(v_1,v_2) = \max_{j_1,j_2 \in [l]} \SW(a_{j_1},b_{j_2})$. Therefore, 
\[
\Pr[W= 1|M=0] \leq \Pr[\SW(v_1,v_2) > \alpha m |M=0] \leq \sum_{j_1,j_2 \in [l] } \Pr[\SW(a_{j_1},b_{j_2})\geq\alpha m|M=0].
\]
Define $X_i = \1_{i \in A_{j_1} \cup B_{j_2}}$. We are going to upper bound $\Pr[\SW(a_{j_1},b_{j_2})\geq \alpha m|M=0]$ in four different cases. 
\begin{enumerate}
\item When $j_1 = J_1$, $j_2 = J_2$: $\SW(a_{j_1},b_{j_2}) = |U_1| = m/2$. In this case $\Pr[\SW(a_{j_1},b_{j_2})\geq\alpha m|M=0] = 0$. 
\item When $j_1 = J_1$ and $j_2 \neq J_2$: Fix $S$ and $T$. We know $X_i = 1$ when $i \in U_1$. $\E[X_i] = 2/3$ when $i \in T \cap S^c$. $\E[X_i] = 1/3$ when $i \in T^c \backslash U_1$. So
\[
\E[\SW(a_{j_1},b_{j_2})] = \E[\sum_{i=1}^m X_i] = m /2 + (2/3) \cdot (m/6) + (1/3) \cdot (m/3) = 13m/18. 
\]
Although $X_i$'s are not independent, but it is easy to check that they are negatively correlated. By generalized chernoff bound for negative correlated random variables, we get
\[
\Pr[\sum_{i=1}^m X_i \geq \alpha m ] \leq exp( -2(\alpha - 13/18)^2 m) \leq e^{-Cm}. 
\]
Thus in this case, we have $Pr[\SW(a_{j_1},b_{j_2})\geq\alpha m|M=0] \leq e^{-Cm}$.
\item When $j_2 = J_2$ and $j_1 \neq J_1$: This is similar to the previous case, we can get $\Pr[\SW(a_{j_1},b_{j_2})\geq\alpha m|M=0] \leq e^{-Cm}$.
\item When $j_1 \neq J_1$ and $j_2 \neq J_2$: Fix $S$ and $T$. We know $\E[X_i] = 8/9$ when $i \in S \cap T$. $\E[X_i] = 7/9$ when $i \in (T \cap S^c) \cup (S \cap T^c)$. $\E[X_i] = 5/9$ when $i \in S^c \cap T^c$. So
\[
\E[\SW(a_{j_1},b_{j_2})] = \E[\sum_{i=1}^m X_i] = 20m/27. 
\]
Although $X_i$'s are not independent, but it is easy to check that they are negatively correlated. By generalized chernoff bound for negative correlated random variables, we get
\[
\Pr[\sum_{i=1}^m X_i \geq \alpha m ] \leq exp( -2(\alpha - 20/27)^2 m) =e^{-Cm}. 
\]
Thus in this case, we have $\Pr[\SW(a_{j_1},b_{j_2})\geq\alpha m|M=0] \leq e^{-Cm}$.
\end{enumerate}
Therefore, we have 
\[
\Pr[W= 1|M=0] \leq l^2 \cdot e^{-Cm} =  e^{-Cm/9}.
\]
To sum up, the auctioneer fails with probability at least
\[
\Pr[M \neq D] - \Pr[W = 1] = 1/2 - 2 e^{-Cm/9}.
\]
\end{enumerate}
\end{proof}


\bibliographystyle{alpha}
\bibliography{bib} 
\end{document}